\documentclass[a4paper,fleqn]{cas-sc}

\usepackage{subfigure}  
\usepackage{graphicx}
\usepackage{booktabs}
\usepackage{enumitem}
\usepackage{float}
\usepackage[section]{placeins}


\usepackage{amsthm}
\usepackage[numbers]{natbib}

\usepackage{algorithm}
\usepackage{algorithmic}	
\usepackage[english]{babel}
\newtheorem{theorem}{Theorem}

\def\tsc#1{\csdef{#1}{\textsc{\lowercase{#1}}\xspace}}
\tsc{WGM}
\tsc{QE}
\tsc{EP}
\tsc{PMS}
\tsc{BEC}
\tsc{DE}

\begin{document}
	\let\WriteBookmarks\relax
	\def\floatpagepagefraction{1}
	\def\textpagefraction{.001}
	\shorttitle{An Efficient and Fast Transformer-Based PINNs}
	\shortauthors{Barman, Chatterjee, Ray}
	
	\title [mode = title]{A Simple but Efficient Transformer-Based Physics-Informed Neural Network for Incompressible Navier--Stokes Equations}                      
	
	
	\author[1]{Biswanath Barman}
    \author[2]{Debdeep Chatterjee}  
        
        \author[1]{Rajendra K. Ray} 
        \cortext[cor1]{Corresponding author}
        \ead{rajendra@iitmandi.ac.in}
        
        \address[1]{School of Mathematical and Statistical Sciences, Indian Institute of Technology Mandi\\
        Mandi, Himachal Pradesh, 175005, India}
        \address[2]{Department of Computer Science \& Engineering, Sikkim Manipal Institute of Technology, \\
        Rangpo, Sikkim, 737136, India}

\begin{abstract}
Traditional approaches for solving fluid dynamics problems often involve high computational cost, mesh sensitivity, and difficulties in accurately capturing complex flow physics. In addition, conventional physics-informed neural networks (PINNs) frequently exhibit degraded performance for strongly nonlinear and temporally evolving systems. In this work, we propose \textit{PhysicsFormer}, a simple and efficient Transformer-based physics-informed neural network framework designed to improve predictive accuracy and computational efficiency for complex flow problems. The proposed architecture incorporates encoder--decoder multi-head attention to capture long-range temporal dependencies and improve information propagation across spatio-temporal domains. Unlike standard multilayer perceptron-based PINNs, \textit{PhysicsFormer} employs pseudo-sequential spatio-temporal representations together with a dynamics-weighted loss formulation for enhanced convergence and stability. Owing to its lightweight architecture and parallel learning strategy, the proposed framework achieves faster training and lower computational cost compared with existing Transformer-based PINN models. The performance of the proposed framework is investigated using the convection equation, Burgers' equation, lid-driven cavity flow at $Re=100$, flow reconstruction and inverse Navier--Stokes problems for flow past a circular cylinder at $Re=100$, and high-Reynolds-number flow reconstruction at $Re=3900$. For the inverse Navier--Stokes problem at $Re=100$, the proposed framework simultaneously reconstructs the flow field and identifies the governing equation parameters with nearly $0\%$ absolute error under both clean and noisy data conditions. Furthermore, for the flow past a circular cylinder at $Re=3900$, \textit{PhysicsFormer} accurately reconstructs the velocity and pressure fields using only $25$ spatial measurements per snapshot over $100$ temporal snapshots. The obtained results demonstrate that \textit{PhysicsFormer} provides an accurate, robust, and computationally efficient framework for complex time-dependent fluid flow simulations.
\end{abstract}
	
	\date{}
	\maketitle
\section{Introduction}
Fluid flow phenomena, including cavity flows, internal pipe flows, and flows past bluff bodies, are commonly governed by the incompressible Navier--Stokes equations, which form a highly nonlinear and strongly coupled system of partial differential equations. To solve these equations, a broad range of computational fluid dynamics (CFD) methodologies has been developed, including finite element, finite volume, and spectral methods \cite{ding2004simulation, liu1996strongly}. Although these numerical techniques have achieved considerable success in fluid mechanics, large-scale simulations often demand substantial computational resources due to the extremely high degrees of freedom associated with practical engineering applications. In addition, CFD solvers may encounter difficulties when dealing with complex geometries, moving interfaces, or dynamically evolving meshes, where numerical stability and convergence become increasingly challenging. To alleviate the computational burden of high-fidelity simulations, reduced-order modeling (ROM) approaches have been introduced as efficient alternatives for approximating dominant flow dynamics \cite{lucia2004reduced}. Several investigations have explored low-dimensional representations of fluid flows using ROM-based techniques \cite{henshaw2007non, jovanovic2014sparsity, hemati2014dynamic, hemmasian2023reduced}. Nevertheless, many conventional ROM approaches rely on linearized or weakly nonlinear assumptions, which may limit their applicability to strongly nonlinear and complex flow phenomena.

With the rapid advancement of scientific machine learning, physics-informed neural networks (PINNs) \cite{lagaris1998artificial, raissi2019physics} have emerged as a promising framework for solving partial differential equations by embedding physical laws directly into deep learning architectures. Conventional PINN formulations \cite{raissi2019physics} and many of their extensions primarily employ multilayer perceptrons (MLPs) to approximate point-wise solutions and have demonstrated encouraging performance across several scientific and engineering applications. Despite these developments, recent studies have identified important limitations of standard PINNs when applied to problems involving high-frequency structures, sharp gradients, or multiscale dynamics \cite{raissi2018deep, fuks2020limitations, krishnapriyan2021characterizing, wang2022and}. In such cases, PINNs often exhibit spectral bias, producing excessively smooth approximations that deviate from the true physical solution, even when the analytical behavior of the governing equations is relatively simple.

To improve the predictive capability of PINNs in these challenging settings, existing approaches generally follow two principal directions. The first strategy introduces additional data regularization or interpolation using information obtained from numerical simulations or experimental measurements \cite{raissi2017physics, zhu2019physics, chen2021physics}. Although such approaches can improve predictive accuracy, they typically depend on the availability of reliable high-fidelity data, which may be computationally expensive or difficult to obtain experimentally. The second strategy focuses on modified optimization procedures and alternative learning methodologies \cite{mao2020physics, krishnapriyan2021characterizing, wang2021understanding, wang2022and}. However, these approaches often increase computational complexity and training cost. For example, the Seq2Seq framework proposed in \cite{krishnapriyan2021characterizing} requires multiple neural networks to be trained sequentially, thereby increasing computational overhead. Similarly, Neural Tangent Kernel (NTK)-based methods \cite{wang2022and} require the construction of kernel matrices $\mathrm{K} \in \mathbb{R}^{\mathrm{D} \times \mathrm{P}}$, where $\mathrm{D}$ denotes the number of training samples and $\mathrm{P}$ represents the number of trainable parameters. Such formulations become increasingly difficult to scale for large datasets and high-dimensional neural network architectures.

Although several studies have attempted to improve the generalization capability and mitigate the failure modes of PINNs \cite{krishnapriyan2021characterizing}, conventional PINN frameworks based on multilayer perceptron architectures often fail to explicitly capture the temporal dependencies inherent in many physical systems. In traditional numerical schemes such as the finite element method \cite{huebner2001finite}, temporal evolution is naturally incorporated through iterative time-marching procedures, where the solution at time $\mathrm{t}+\Delta \mathrm{t}$ directly depends on the state at time $\mathrm{t}$. In contrast, standard PINNs generally operate in a point-wise learning framework and do not inherently model the sequential temporal structure of partial differential equations. Consequently, the propagation of initial-condition information across the temporal domain may become insufficient, particularly for complex dynamical systems. This limitation frequently leads to failure modes in which the predicted solutions remain accurate near the initial state but gradually deteriorate over time, producing overly smooth approximations that capture primarily low-frequency components while failing to resolve high-frequency solution structures \cite{krishnapriyan2021characterizing}.

In parallel with PINN-based developments, data-driven operator learning approaches have also attracted considerable attention for solving partial differential equations. These methods utilize neural operators to learn nonlinear mappings between input and output function spaces directly from data, enabling rapid inference once training is completed. Among such approaches, the Fourier Neural Operator (FNO) \cite{li2020fourier} has demonstrated strong capability in predicting high-resolution solutions from coarse input representations. Despite their computational efficiency during inference, these approaches generally require large training datasets and often do not explicitly incorporate the governing physical laws into the learning process. As a result, important physical characteristics embedded within the PDE system may not be fully preserved.

A natural direction for overcoming the limited treatment of temporal dependencies in conventional PINNs is the adoption of Transformer-based architectures, which are well known for their ability to model long-range interactions in sequential data through self-attention and encoder--decoder mechanisms \cite{vaswani2017attention}. Transformer variants have achieved remarkable success in numerous application domains; however, adapting architectures originally designed for sequential learning to physics-informed frameworks remains challenging. Key difficulties arise in the representation of spatio-temporal data and the construction of physically consistent loss formulations. Recent studies, such as PINNsFormer \cite{zhao2023pinnsformer}, have explored Transformer-based PINN frameworks for PDE learning. Nevertheless, these architectures often require substantial GPU memory and computational resources for complex problems, limiting their practical applicability in resource-constrained environments. Motivated by these challenges, we introduce an efficient Transformer-based physics-informed neural network framework, termed \textit{PhysicsFormer}, designed to reduce computational overhead while maintaining high predictive accuracy for complex fluid flow problems.

The proposed \textit{PhysicsFormer} transforms conventional point-wise physics-informed learning into a sequential spatio-temporal learning framework through encoder--decoder multi-head cross-attention. Owing to its lightweight and parallel architecture, the proposed framework achieves faster convergence and reduced computational cost compared with existing Transformer-based PINN models. In addition, a weighted sinusoidal activation function together with a dynamic loss-weighting strategy is introduced to improve the representation of high-frequency flow structures and enhance training stability. The effectiveness of the proposed framework is investigated using several benchmark problems, including the convection equation, Burgers' equation, lid-driven cavity flow at $Re=100$, flow reconstruction and inverse Navier--Stokes problems for flow past a circular cylinder at $Re=100$, and high-Reynolds-number flow reconstruction at $Re=3900$. For the lid-driven cavity benchmark at $Re=100$, the predicted velocity and pressure fields show strong agreement with reference numerical solutions with minimal reconstruction error. In the inverse Navier--Stokes problem at $Re=100$, the proposed framework identifies the unknown parameters with nearly $0\%$ absolute error under both clean and noisy data conditions. Furthermore, for the flow past a circular cylinder at $Re=3900$, \textit{PhysicsFormer} accurately reconstructs the velocity and pressure fields from sparse measurement data while preserving the complex wake structures across different temporal snapshots. These results demonstrate that the proposed framework can efficiently and accurately capture complex spatio-temporal flow dynamics.

\textbf{Main Contributions:}
\begin{itemize}

    \item We propose \textit{PhysicsFormer}, a simple and efficient Transformer-based physics-informed neural network that incorporates encoder--decoder multi-head cross-attention to improve temporal dependency learning in nonlinear flow problems.

    \item A lightweight architecture together with a trainable weighted sinusoidal activation function, $w\sin(t)$, and a dynamic loss-weighting strategy is introduced to improve convergence, high-frequency feature representation, and computational efficiency.

    \item The failure modes of conventional PINNs are analyzed using the convection equation, highlighting the limitations of point-wise learning for temporally evolving systems.

    \item The proposed framework is validated on several benchmark problems, including the Burgers equation, lid-driven cavity flow at $Re=100$, flow reconstruction and inverse Navier--Stokes problems for flow past a circular cylinder at $Re=100$, and high-Reynolds-number flow reconstruction at $Re=3900$.

    \item \textit{PhysicsFormer} accurately reconstructs velocity, pressure, vorticity, and streamline fields from sparse measurements, while the inverse Navier--Stokes framework identifies unknown parameters with nearly $0\%$ absolute error under both clean and noisy data conditions.

\end{itemize}
The remainder of this paper is organized as follows. Section~\ref{sec:relatedworks} reviews recent developments in PINNs, including approaches proposed to overcome the failure modes of conventional PINNs and recent Transformer-based formulations. Section~\ref{sec:preliminary} presents the preliminaries of PINNs together with the proposed activation function. The formulation and methodology of the proposed \textit{PhysicsFormer} framework are described in Section~\ref{sec:methodology}. Section~\ref{sec:numerical_study} discusses the benchmark problems, numerical experiments, and flow reconstruction studies. Finally, conclusions and future directions are summarized in Section~\ref{sec:conclusion}. Additional hyperparameter sensitivity analyses of the proposed framework are provided in Appendices~A--C.

\section{Related Works}
\label{sec:relatedworks}
\textbf{Traditional Numerical Methods:} 
Obtaining analytical solutions for partial differential equations (PDEs) remains a fundamental challenge in many areas of scientific computing. Consequently, PDEs are commonly discretized over computational meshes and solved using numerical approaches such as finite difference methods \cite{sod1978survey}, finite element methods \cite{huebner2001finite}, and spectral methods \cite{bernardi1997spectral}. Although these techniques have achieved considerable success in modeling complex physical systems, they often require substantial computational resources and long simulation times for large-scale or highly nonlinear problems \cite{umetani2018learning}. In addition, conventional numerical solvers may exhibit limited efficiency when applied to inverse problems involving parameter estimation or incomplete observational data.

\textbf{Physics-informed Neural Networks:} 
Physics-informed neural networks (PINNs) \cite{raissi2019physics} have emerged as an effective deep learning framework for solving partial differential equations without relying on large labeled datasets. In this methodology, the governing equations together with the associated initial and boundary conditions are embedded directly into the loss function of the neural network \cite{raissi2019physics, yu2018deep}. By minimizing the residuals of the underlying physical laws during training, the network learns solutions that remain consistent with the prescribed physics. Owing to this capability, PINNs have attracted considerable attention across a broad range of scientific and engineering applications, including fluid dynamics, solid mechanics, and quantum systems \cite{carleo2019machine, yang2020physics}. In recent years, several studies have focused on improving the learning behavior and predictive performance of PINNs through enhanced training strategies and architectural modifications \cite{mao2020physics, wang2021understanding, wang2022and, wang2022auto, khademi2025physics, farea2025learnable}. These developments have contributed to improvements in convergence, generalization capability, and model interpretability \cite{cuomo2022scientific}. Nevertheless, conventional PINNs still encounter difficulties when dealing with strongly nonlinear flows, long-range temporal interactions, and multiscale physical phenomena. In particular, the direct application of standard PINNs may become inadequate for stiff or highly complex PDE systems. To overcome these challenges, modified formulations such as asymptotic-preserving PINNs have been proposed for multiscale kinetic and transport equations \cite{jin2023asymptotic, liu2025asymptotic}.

\textbf{Operator Learning:} 
Operator learning has emerged as an important branch of data-driven deep learning, where neural operators are trained to approximate solution operators associated with partial differential equations. This class of methods is well suited to mapping input fields to output fields in a variety of physical settings, including the prediction of future fluid states from past observations and the estimation of internal stress distributions in solid mechanics \cite{lu2021learning}. Representative models in this area include the Fourier Neural Operator \cite{li2020fourier} and several related extensions, such as those reported in \cite{li2023geometry, rahman2022u}. A number of recent studies \cite{yin2022continuous} have also focused on learning the temporal evolution of PDE systems. However, many of these approaches depend on large training datasets and do not always incorporate the underlying physical structure of the governing equations in an explicit way.

\textbf{PINN Failure Modes:} 
Although physics-informed neural networks (PINNs) \cite{raissi2019physics} have demonstrated considerable promise for solving partial differential equations, several recent studies have identified important limitations associated with their training and predictive capabilities, particularly for problems involving high-frequency components, multiscale dynamics, and strongly nonlinear behavior \cite{fuks2020limitations, raissi2018deep, braga2021self, krishnapriyan2021characterizing, zhao2023pinnsformer, wang2022and}. These challenges have motivated the development of alternative strategies, including improved network architectures, modified learning frameworks, and enhanced data-driven formulations \cite{han2018solving, lou2021physics, wang2021understanding, wang2022and}. A deeper understanding of the underlying failure mechanisms of PINNs is therefore essential for the reliable simulation of complex physical systems. One commonly reported limitation is the difficulty of conventional PINNs in accurately representing high-frequency solution structures and temporally evolving chaotic flow dynamics.

\textbf{Transformer-Based Models:} 
The Transformer architecture \cite{vaswani2017attention} has received significant attention due to its ability to model long-range dependencies through self-attention mechanisms, leading to major advances in natural language processing applications \cite{kalyan2021ammus}. Owing to its strong sequence-learning capability, the Transformer framework has subsequently been extended to several other domains, including computer vision, speech recognition, and time-series forecasting \cite{dong2018speech, han2022survey, wen2022transformers}. In recent years, growing interest has emerged in applying Transformer-based models to problems governed by partial differential equations \cite{cao2021choose, wu2024transolver, zhao2023pinnsformer, zhu2026physicssolver, zhaopinnsformer}. Despite these developments, the integration of Transformer architectures within physics-informed neural network frameworks remains relatively unexplored, particularly for complex fluid flow applications. Furthermore, there is still a need for computationally efficient Transformer-based PINN architectures capable of accurately learning nonlinear spatio-temporal dynamics associated with PDE systems.

\section{Preliminary}
\label{sec:preliminary}
\subsection{Physics-Informed Neural Networks}
Let us examine the initial-boundary value problem:

\begin{align}
\mathcal{D}_{\mathbf{x}, \mathrm{t}}[\boldsymbol{u}(\mathbf{x}, \mathrm{t})] &= \mathrm{f}(\mathbf{x}, \mathrm{t}), && \mathbf{x} \in \Omega, \; \mathrm{t} \in (0, T] 
\label{eq:pde}\\
\mathcal{B}_{\mathbf{x}, \mathrm{t}}[\boldsymbol{u}(\mathbf{x}, \mathrm{t})] &= \mathrm{g}(\mathbf{x}, \mathrm{t}), && \mathbf{x} \in \partial \Omega, \; \mathrm{t} \in (0, T] 
\label{eq:bc}\\
\boldsymbol{u}(\mathbf{x}, 0) &= \mathrm{h}(\mathbf{x}), && \mathbf{x} \in \bar{\Omega}
\label{eq:ic}
\end{align}

Let $\Omega \subset \mathbb{R}^d$ be an open set, with $\bar{\Omega}$ representing its closure. The function $\boldsymbol{u}: \bar{\Omega} \times [0, T] \rightarrow \mathbb{R}$ denotes the required solution, where $\mathbf{x} \in \Omega$ is a spatial vector variable and $\mathrm{t}$ signifies time. The operators $\mathcal{D}_{\mathbf{x}, \mathrm{t}}$ and $\mathcal{B}_{\mathbf{x}, \mathrm{t}}$ are spatial-temporal differential operators. The problem data consists of the forcing function $f: \Omega \rightarrow \mathbb{R}$, the boundary condition function $g: \partial \Omega \times (0, T]$, and the initial condition function $h: \bar{\Omega} \rightarrow \mathbb{R}$. Moreover, sensor data within the interior of the domain may be accessible. We presume that the data are enough and suitable for a well-defined problem. Time-independent problems and other data types can be addressed in a similar manner; thus, we will utilize equations ~\eqref{eq:pde}–\eqref{eq:ic} as a framework. According to \cite{raissi2019physics}, let $\boldsymbol{u}(\mathbf{x}, \mathrm{t})$ be represented by the output $\boldsymbol{u}(\mathbf{x}, \mathrm{t} ; \mathbf{w})$ of a deep neural network, with inputs $\mathbf{x}$ and $\mathrm{t}$ (in the context of a PDE system, this would entail a neural network with many outputs).

The value of $\mathcal{D}_{\mathbf{x}, \mathrm{t}}[\boldsymbol{u}(\mathbf{x}, \mathrm{t} ; \mathbf{w})]$ and $\mathcal{B}_{\mathbf{x}, \mathrm{t}}[\boldsymbol{u}(\mathbf{x}, \mathrm{t} ; \mathbf{w})]$ can be calculated quickly and accurately via reverse-mode automatic differentiation \cite{baydin2018automatic}. The network weights $\mathbf{w}$ are optimized by minimizing a loss function that penalizes the output for failing to meet conditions ~\eqref{eq:pde}–\eqref{eq:ic}:

\begin{equation}
\mathcal{L}_{\text{PINNs}}(\mathbf{w})=
\mathcal{L}_{data}(\mathbf{w})+
\mathcal{L}_{residual}(\mathbf{w})+
\mathcal{L}_{bc}(\mathbf{w})+
\mathcal{L}_{ic}(\mathbf{w})
\label{eq:pinns_loss}
\end{equation}

where $\mathcal{L}_{data}$ denotes the loss term associated with sample data (if applicable), whereas $\mathcal{L}_{residual}$, $\mathcal{L}_{bc}$, and $\mathcal{L}_{ic}$ represent the loss terms related to the violation of the PDE ~\eqref{eq:pde}, the boundary condition ~\eqref{eq:bc}, and the initial condition ~\eqref{eq:ic}, respectively:

\begin{equation}
\begin{aligned}
\mathcal{L}_{data}(\mathbf{w}) &= \frac{1}{\mathcal{N}_d} \sum_{i=1}^{\mathcal{N}_d}\left|\boldsymbol{u}\left(\mathbf{x}_d^i, \mathrm{t}_d^i ; \mathbf{w}\right)-y_d^i\right|^2 \\
\mathcal{L}_{residual}(\mathbf{w}) &= \frac{1}{\mathcal{N}_r} \sum_{i=1}^{\mathcal{N}_r}\left|\mathcal{D}_{\mathbf{x}, \mathrm{t}}\left[\boldsymbol{u}\left(\mathbf{x}_r^i, t_r^i ; \mathbf{w}\right)\right]-f\left(\mathbf{x}_r^i, \mathrm{t}_r^i\right)\right|^2 \\
\mathcal{L}_{bc}(\mathbf{w}) &= \frac{1}{\mathcal{N}_b} \sum_{i=1}^{\mathcal{N}_b}\left|\mathcal{B}_{\mathbf{x}, \mathrm{t}}\left[\boldsymbol{u}\left(\mathbf{x}_b^i, \mathrm{t}_b^i ; \mathbf{w}\right)\right]-g\left(\mathbf{x}_b^i, \mathrm{t}_b^i\right)\right|^2 \\
\mathcal{L}_{ic}(\mathbf{w}) &= \frac{1}{\mathcal{N}_0} \sum_{i=1}^{\mathcal{N}_0}\left|\boldsymbol{u}\left(\mathbf{x}_0^i, 0 ; \mathbf{w}\right)-h\left(\mathbf{x}_0^i\right)\right|^2
\end{aligned}
\end{equation}

Let $\left\{\mathbf{x}_d^i, \mathrm{t}_d^i, y_d^i\right\}_{i=1}^{\mathcal{N}_d}$ represent sensor data (if available), $\left\{\mathbf{x}_0^i\right\}_{i=1}^{\mathcal{N}_0}$ denote initial condition points, $\left\{\mathbf{x}_b^i, \mathrm{t}_b^i\right\}_{i=1}^{\mathcal{N}_b}$ signify boundary condition points, and $\left\{\mathbf{x}_r^i, \mathrm{t}_r^i\right\}_{i=1}^{\mathcal{N}_r}$ indicate residual ("collocation") points randomly distributed within the domain $\Omega$. Here, $\mathcal{N}_d, \mathcal{N}_0, \mathcal{N}_b$, and $\mathcal{N}_r$ correspond to the total counts of sensor, initial, boundary, and residual points, respectively. The network weights $\mathbf{w}$ can be optimized by minimizing the overall training loss $\mathcal{L}_{\text {PINNs}}(\mathbf{w})$ by conventional gradient descent methods \cite{kingma2014adam} employed in deep learning. A schematic illustration of the PINNs framework is presented in Figure.~\ref{fig:pinns_flowchart}.

\begin{theorem}[Universal Approximation Theorem] \cite{hornik1991approximation} \label{thm:universal_approx}
 Let $\sigma$ denote a continuous, bounded, and non-constant activation function. For any continuous function $f$ defined on a compact subset $\mathrm{K} \subset \mathbb{R}^n$ and for any $\epsilon>0$, there exists a feedforward neural network with a single hidden layer and a limited number of neurons such that the network's output $\hat{\mathrm{f}}$ approximates $\mathrm{f}$ to within $\epsilon$, i.e.,

$$
|\mathrm{f}(\mathbf{x})-\hat{\mathrm{f}}(\mathbf{x})|<\epsilon \quad \forall \mathbf{x} \in\mathrm{K}
$$

This theorem states that a feedforward neural network, equipped with at least one hidden layer containing a sufficient number of neurons, may approximate any continuous function to an arbitrary degree of accuracy.
\end{theorem}

\begin{theorem} \cite{zhao2023pinnsformer, zhaopinnsformer}
Let $\mathcal{N}$ represent a one-hidden-layer neural network of infinite width, utilizing the activation function $\phi(\mathrm{t}) = w\sin(\mathrm{t})$; thus, $\mathcal{N}$ acts as a universal approximator for any real-valued target function $\mathrm{f}$.
\end{theorem}

\begin{proof}
Let $\phi:\mathbb{R}\to\mathbb{R}$ be defined by $\phi(t)=w\sin t$, where $w\in\mathbb{R}$ is a trainable scalar parameter with $w\neq 0$.

\emph{Continuity.} The sine function $\sin t$ is continuous on $\mathbb{R}$, and multiplication by a constant preserves continuity. Therefore, $\phi(t)=w\sin t$ is continuous for all $t\in\mathbb{R}$.

\emph{Boundedness.} Since $|\sin t|\le 1$ for every $t\in\mathbb{R}$, we obtain
\[
|\phi(t)| = |w||\sin t| \le |w|.
\]
Hence, $\phi(t)$ is bounded on $\mathbb{R}$.

\emph{Non-constancy.} The function $\sin t$ is non-constant, as $\sin 0 = 0$ and $\sin(\pi/2) = 1$. Thus, for any nonzero $w$, the function $\phi(t)=w\sin t$ is also non-constant.

Since $\phi(t)$ is continuous, bounded, and non-constant, it satisfies the conditions required by Theorem~\ref{thm:universal_approx}. Therefore, neural networks employing $\phi(t)=w\sin t$ as the activation function possess the universal approximation property; that is, such networks can uniformly approximate any continuous function on a compact subset of $\mathbb{R}^d$ to arbitrary accuracy.
\end{proof}
\refstepcounter{figure}

\begin{center}
    \includegraphics[width=0.7\textwidth]{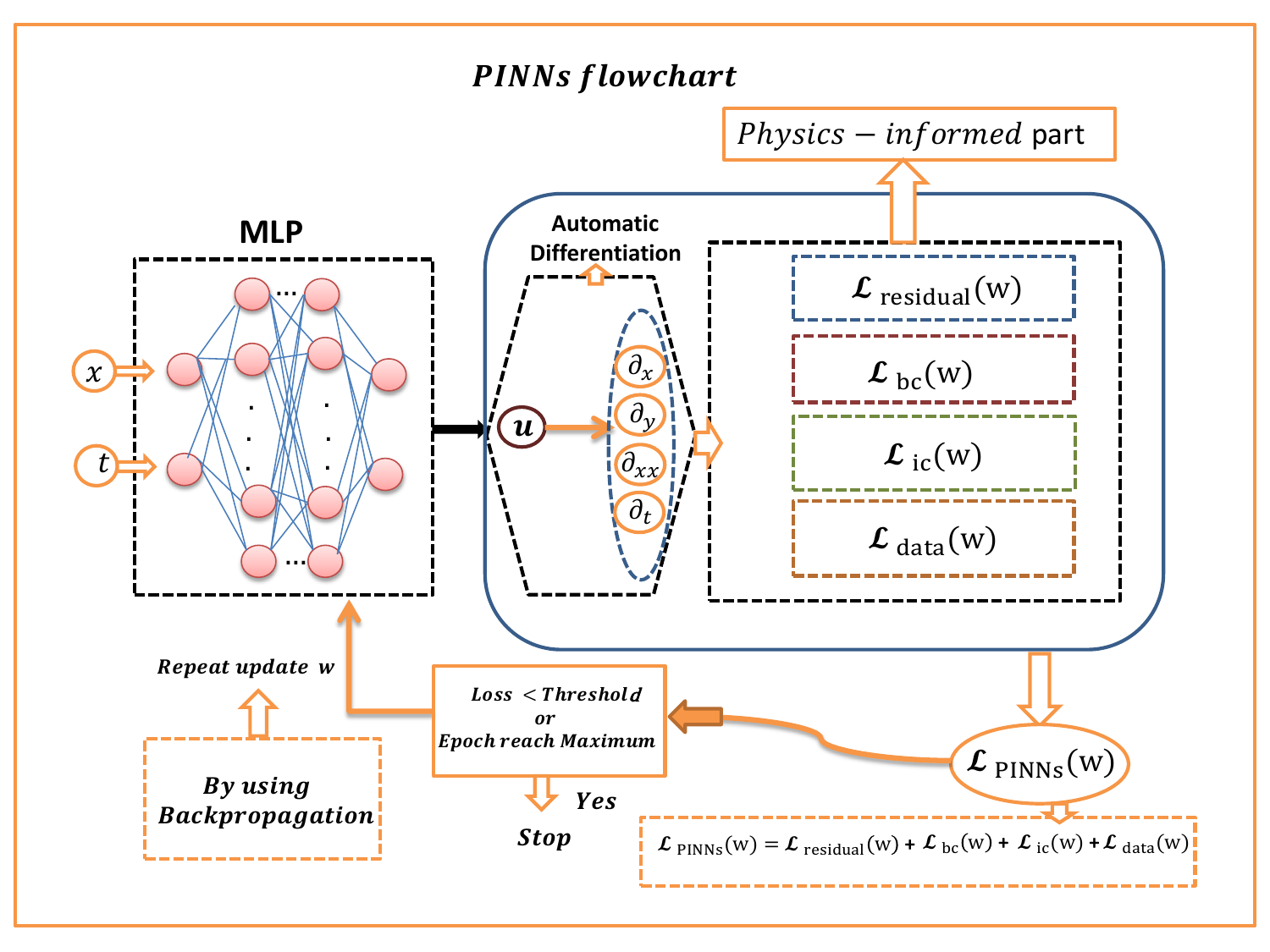}
\end{center}

\begin{quote}
\noindent\textbf{Figure \thefigure.}
Schematic representation of the Physics-Informed Neural Network (PINN) framework for solving partial differential equations.
\end{quote}

\label{fig:pinns_flowchart}

\section{Methodology}
\label{sec:methodology}
Conventional physics-informed neural networks are mainly developed for point-wise prediction and therefore do not explicitly capture the temporal relationships present in many physical systems. In standard PINN formulations, the spatial coordinates $\mathbf{x}$ and temporal variable $\mathrm{t}$ are treated as independent inputs to approximate the solution $u(\mathbf{x},\mathrm{t})$, which may limit their ability to model time-dependent dynamics accurately, particularly for hyperbolic and parabolic partial differential equations. To address this limitation, we propose a Transformer-based physics-informed neural network framework, termed \textit{PhysicsFormer}, which reformulates PDE learning as a sequential spatio-temporal prediction problem. Inspired by the recently developed PINNsFormer framework \cite{zhao2023pinnsformer, zhaopinnsformer}, the proposed architecture combines multi-head encoder--decoder attention with physics-informed learning while significantly reducing computational cost and GPU memory usage. The framework consists of a data embedding module, high-dimensional feature generation, an encoder--decoder attention mechanism, and an output prediction layer. In addition, a weighted sinusoidal activation function, $\phi(\mathrm{t}) = w\sin(\mathrm{t})$, is introduced to improve the representation of high-frequency solution structures while maintaining computational efficiency. A schematic illustration of the proposed framework is shown in Figure~\ref{fig:physicsformer_flowchart}.

\subsection{Data-Embedder}
Conventional physics-informed neural networks operate on point-wise spatio-temporal inputs, whereas Transformer architectures are inherently designed to learn long-range dependencies from sequential data. To bridge this difference, the spatio-temporal inputs are reformulated into temporally ordered pseudo-sequences before being processed by the network. For a spatial coordinate $\mathbf{x} \in \mathbb{R}^{d-1}$ and temporal input $\mathrm{t} \in \mathbb{R}$, the Data-Embedder generates a sequence of neighboring temporal states as
$$
[\mathbf{x}, \mathrm{t}] \stackrel{\text{Embedder}}{\Longrightarrow} \{[\mathbf{x}, \mathrm{t}], [\mathbf{x}, \mathrm{t}+\Delta \mathrm{t}], \ldots, [\mathbf{x}, \mathrm{t}+(k-1)\Delta \mathrm{t}]\}.
$$
Here, $[\cdot]$ denotes the concatenation operation, producing vectorized inputs in $\mathbb{R}^{d}$, while the generated pseudo-sequence is represented in $\mathbb{R}^{k \times d}$. In this framework, a single spatio-temporal point is transformed into multiple temporally correlated states, enabling the network to learn sequential physical dependencies more effectively. The hyperparameters $k$ and $\Delta \mathrm{t}$ determine the sequence length and temporal step size, respectively. In practice, excessively large values of $k$ increase computational and memory costs, whereas large $\Delta \mathrm{t}$ values may weaken the temporal correlation between adjacent sequence elements.

\subsection{High-Dimensional Space Generation}
In addition to the Data-Embedder, the proposed \textit{PhysicsFormer} framework consists of three primary components: a high-dimensional feature generation module, a multi-head attention encoder--decoder, and an output layer. The output layer is implemented as a fully connected multilayer perceptron attached to the final stage of the decoder. Detailed descriptions of the first two components are provided in the following subsections. Unlike architectures that rely on convolutional or recurrent operations, \textit{PhysicsFormer} employs only linear transformations and nonlinear activation functions, resulting in a lightweight and computationally efficient framework. Compared with PINNsFormer \cite{zhao2023pinnsformer, zhaopinnsformer}, the proposed architecture achieves improved computational efficiency while benefiting from the parallel learning capability of Transformer-based models, thereby avoiding the sequential processing limitations commonly associated with recurrent neural networks. Since many partial differential equations involve low-dimensional spatial and temporal variables, directly supplying these inputs to the encoder may not sufficiently capture complex feature interactions. To address this limitation, the proposed framework projects the sequential spatio-temporal inputs into a higher-dimensional latent space using a fully connected multilayer perceptron. Unlike conventional embedding strategies that primarily encode semantic similarity \cite{vaswani2017attention, devlin2019bert}, the proposed high-dimensional space generation module constructs enriched feature representations by combining all spatio-temporal components through a linear projection process, thereby improving the expressive capability of the network.

\subsection{Encoder-Decoder Architecture}
The proposed \textit{PhysicsFormer} adopts an encoder--decoder architecture inspired by the original Transformer framework \cite{vaswani2017attention}. The encoder is composed of multiple identical layers, each containing a self-attention mechanism followed by a feed-forward network. The decoder structure is slightly modified compared with the standard Transformer architecture, where each layer consists of an encoder--decoder cross-attention module and a feed-forward layer. In the proposed framework, the same spatio-temporal embeddings are supplied to both the encoder and decoder, eliminating the need for the decoder to reconstruct dependencies from separate embedded representations. A schematic illustration of the encoder--decoder structure of \textit{PhysicsFormer} is presented in Figure~\ref{fig:physicsformer_flowchart}. The self-attention mechanism within the encoder captures global dependencies among the spatio-temporal inputs, while the cross-attention module in the decoder selectively emphasizes relevant physical relationships during the reconstruction process, thereby improving feature extraction compared with conventional PINN architectures. Unlike sequence prediction tasks in natural language processing or time-series forecasting, physics-informed neural networks primarily focus on approximating the physical state at a given instant. Consequently, identical embeddings are employed for both the encoder and decoder to preserve the consistency of the learned physical representations.

\subsection{Proposed PhysicsFormer} 
While classical physics-informed neural networks focus on point-to-point predictions, their application to pseudo-sequential inputs has not been explored. In \textit{PhysicsFormer}, each point produced in the ordered sequence is referred to as$\left[\mathbf{x}_i, \mathrm{t}_i+\gamma \Delta \mathrm{t}\right]$, is associated with the corresponding approximation, $\hat{\boldsymbol{u}}\left(\mathbf{x}_i, \mathrm{t}_i+\gamma \Delta \mathrm{t}\right)$ for any $\gamma \in \mathbb{N}$, where $\gamma < k$. This method enables the computation of the $n$th-order gradients concerning $\mathbf{x}$ or $\mathrm{t}$ independently for any permissible $n$. For example, for each specified input pseudo sequence $\left\{\left[\mathbf{x}_i, \mathrm{t}_i\right],\left[\mathbf{x}_i, \mathrm{t}_i+\Delta \mathrm{t}\right], \ldots,\left[\mathbf{x}_i, \mathrm{t}_i+\right.\right. (k-1) \Delta \mathrm{t}]\}$, and the relevant approximations $\left\{\hat{\boldsymbol{u}}\left(\mathbf{x}_i, \mathrm{t}_i\right), \hat{\boldsymbol{u}}\left(\mathbf{x}_i, \mathrm{t}_i+\Delta \mathrm{t}\right), \ldots, \hat{\boldsymbol{u}}\left(\mathbf{x}_i, \mathrm{t}_i+(k-\right.\right.$ 1) $\Delta \mathrm{t})\}$, we may calculate the first-order derivatives with respect to $\mathbf{x}$ and $\mathrm{t}$ independently as follows:

\begin{equation}
\begin{aligned}
& \frac{\partial\left\{\hat{\boldsymbol{u}}\left(\mathbf{x}_i, \mathrm{t}_i+\gamma \Delta \mathrm{t}\right)\right\}_{j=0}^{k-1}}
       {\partial\left\{\mathrm{t}_i+j \Delta \mathrm{t}\right\}_{j=0}^{k-1}}
   = \left\{\frac{\partial \hat{\boldsymbol{u}}\left(\mathbf{x}_i, \mathrm{t}_i\right)}{\partial \mathrm{t}_i},
     \frac{\partial \hat{\boldsymbol{u}}\left(\mathbf{x}_i, \mathrm{t}_i+\Delta \mathrm{t}\right)}{\partial\left(\mathrm{t}_i+\Delta \mathrm{t}\right)},
     \ldots,
     \frac{\partial \hat{\boldsymbol{u}}\left(\mathbf{x}_i, \mathrm{t}_i+(k-1) \Delta \mathrm{t}\right)}
          {\partial\left(\mathrm{t}_i+(k-1) \Delta \mathrm{t}\right)}\right\} \\[6pt]
& \frac{\partial\left\{\hat{\boldsymbol{u}}\left(\mathbf{x}_i, \mathrm{t}_i+\gamma \Delta \mathrm{t}\right)\right\}_{j=0}^{k-1}}
       {\partial \boldsymbol{x}_i}
   = \left\{\frac{\partial \hat{\boldsymbol{u}}\left(\boldsymbol{x}_i, \mathrm{t}_i\right)}{\partial \boldsymbol{x}_i},
     \frac{\partial \hat{\boldsymbol{u}}\left(\mathbf{x}_i, \mathrm{t}_i+\Delta \mathrm{t}\right)}{\partial \boldsymbol{x}_i},
     \ldots,
     \frac{\partial \hat{\boldsymbol{u}}\left(\mathbf{x}_i, \mathrm{t}_i+(k-1) \Delta \mathrm{t}\right)}{\partial \mathbf{x}_i}\right\}
\end{aligned}
\label{eq:time-space-derivatives}
\end{equation}

This technique can be easily extended to higher-order derivatives to compute the gradients of sequential approximations in relation to sequential inputs. It refers to residual, boundary, and initial conditions. In contrast to the conventional optimization objective of PINNs in Equation~\eqref{eq:pinns_loss}, which incorporates initial and boundary condition objectives, \textit{PhysicsFormer} distinguishes between the two and uses separate regularization techniques for initial and boundary conditions within its learning framework. The PINNs loss regularizes all sequential outputs for boundary points and residuals, employing a weighted loss to adjust each loss component. This occurs because each generated pseudo-timestep resides inside the same domain as the actual inputs. For example, if $\left[\mathbf{x}_i, \mathrm{t}_i\right]$ is derived from the boundary, then $\left[\mathbf{x}_i, \mathrm{t}_i+\gamma \Delta \mathrm{t}\right]$ also exists at the boundary for any $\gamma \in \mathbb{N}^{+}$. Conversely, for initial points, just $\mathrm{t}=0$ is considered. The condition is regularized with respect to the initial element in the consecutive outputs. This is because only the initial element of the pseudo-sequence accurately fulfills the primary condition at $t=0$. The subsequent time steps are defined as $\mathrm{t}=\gamma \Delta \mathrm{t}$ for any $\gamma \in \mathbb{N}^{+}$, extending beyond the initial conditions.

\refstepcounter{figure}

\begin{center}
    \includegraphics[width=\linewidth]{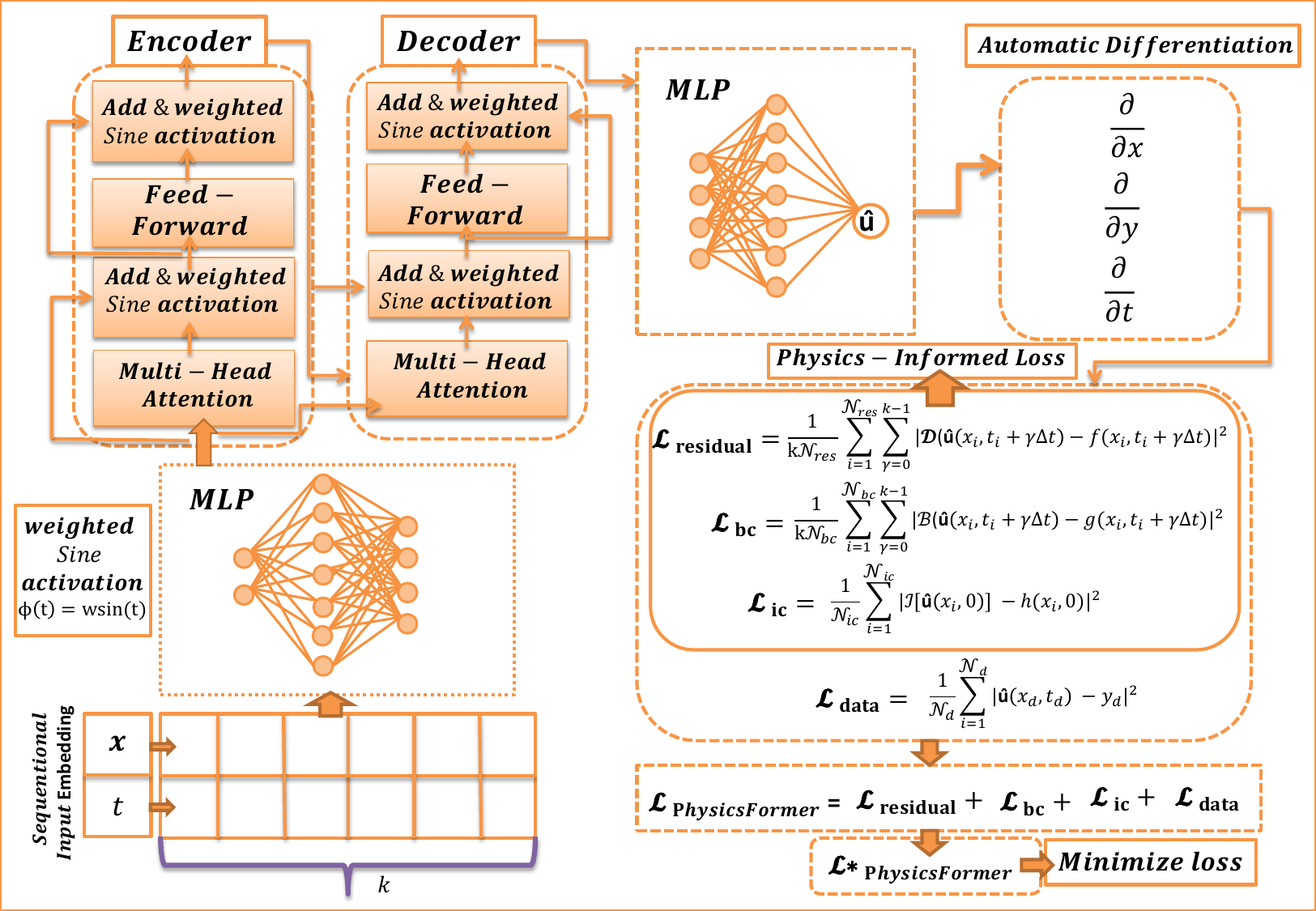}
\end{center}

\begin{quote}
\noindent\textbf{Figure \thefigure.}
Flowchart of \textit{PhysicsFormer} for solving general partial differential equations (PDEs).
\end{quote}

\label{fig:physicsformer_flowchart}
Considering these parameters, we adjust the loss function of the PINNs for the sequential version, as detailed below:

\begin{equation}
\begin{gathered}
\mathcal{L}_{\text {residual}}=\frac{1}{k \mathcal{N}_{\text {res}}} 
\sum_{i=1}^{\mathcal{N}_{\text {res }}} \sum_{\gamma=0}^{k-1}
\left|\mathcal{D}\left[\hat{\boldsymbol{u}}\left(\mathbf{x}_i, \mathrm{t}_i+\gamma \Delta \mathrm{t}\right)\right]
-f\left(\mathbf{x}_i, \mathrm{t}_i+\gamma \Delta \mathrm{t}\right)\right|^2 \\[6pt]
\mathcal{L}_{bc}=\frac{1}{k \mathcal{N}_{bc}} 
\sum_{i=1}^{\mathcal{N}_{bc}} \sum_{\gamma=0}^{k-1}
\left|\mathcal{B}\left[\hat{\boldsymbol{u}}\left(\mathbf{x}_i, \mathrm{t}_i+\gamma \Delta \mathrm{t}\right)\right]
-g\left(\mathbf{x}_i, \mathrm{t}_i+\gamma \Delta \mathrm{t}\right)\right|^2 \\[6pt]
\mathcal{L}_{ic}=\frac{1}{\mathcal{N}_{ic}} 
\sum_{i=1}^{\mathcal{N}_{ic}}
\left|\mathcal{I}\left[\hat{\boldsymbol{u}}\left(\mathbf{x}_i, 0\right)\right]
-h\left(\mathbf{x}_i, 0\right)\right|^2 \\[6pt]
\mathcal{L}_{data}=\frac{1}{\mathcal{N}_d} 
\sum_{i=1}^{\mathcal{N}_d}\left|u\left(\mathbf{x}_d^i, t_d^i \right)-y_d^i\right|^2 \\[6pt]
\mathcal{L}_{\text {PhysicsFormer }}=
\lambda_{\text {residual}} \mathcal{L}_{\text {residual}}
+\lambda_{ic} \mathcal{L}_{ic}
+\lambda_{bc} \mathcal{L}_{bc} 
+\lambda_{data}\mathcal{L}_{data} 
\end{gathered}
\label{eq:loss}
\end{equation}

where $\mathcal{N}_{res}=\mathcal{N}_r$ denotes the residual points as specified in Equation~\eqref{eq:pinns_loss}, $\mathcal{N}_{bc}, \mathcal{N}_{ic}$ denote the quantity of boundary and initial points, respectively, with $\mathcal{N}_{bc}+\mathcal{N}_{ic}=\mathcal{N}_b$. In which location, $\left\{\mathbf{x}_d^i, \mathrm{t}_d^i, y_d^i\right\}_{i=1}^{\mathcal{N}_d}$ denotes experimental data (if available), $\left\{\mathbf{x}_0^i\right\}_{i=1}^{\mathcal{N}_0}$ represents initial condition points. Similar to the PINNs loss, the regularization weights $\lambda_{res}$, $\lambda_{bc}$, $\lambda_{ic}$, and $\lambda_{data}$ balance the importance of the loss terms in \textit{PhysicsFormer}.

During the training process, all initial, boundary, and residual points are conveyed through \textit{PhysicsFormer} to acquire appropriate sequential approximations. The augmented PINNs loss $\mathcal{L}_{\text {PhysicsFormer }}$ in Equation~\eqref{eq:loss} is subsequently improved using gradient-based approaches, including L-BFGS, Adam, or a combination of both. The model parameters are subsequently adjusted until convergence is attained. The sequential solutions are assessed during the testing phase by \textit{PhysicsFormer} sending any arbitrary pair $[\mathbf{x}, \mathrm{t}]$. The initial portion of the sequential approximation precisely corresponds to the following value of $\hat{\boldsymbol{u}}(\mathbf{x}, \mathrm{t})$.

\subsection{Description of Cylinder Wake Data} The resolution of incompressible flow around a circular cylinder is a fundamental challenge in fluid mechanics. This research utilizes data from two-dimensional wake flow around a circular cylinder at a low Reynolds number of $\operatorname{Re}=u_{\infty} D / \nu=100$. The nondimensional free stream velocity is assumed to be $u_{\infty}=1$, the cylinder diameter is $D=1$, and the kinematic viscosity is $\nu=0.01$. The system attains a periodic steady flow state exhibiting an asymmetrical Kármán vortex street \cite{raissi2017physics} in the wake, as illustrated in Figure ~\ref{fig:physicsformer_supervised}. The grid and velocity are nondimensionalized utilizing the free stream velocity $u_{\infty}$ and the cylinder diameter $D$. Boundary conditions consist of a uniform free flow velocity on the left, a zero-pressure outlet on the right boundary situated 25 diameters downstream, and symmetric conditions $[-15,25] \times[-8,8]$ at the upper and lower boundaries of the domain. The dataset was obtained using direct numerical simulation of the Navier-Stokes equations; for further details, refer to this publication \cite{raissi2019physics, raissi2017physics}.

For the purpose of simplification, a rectangle region downstream of the cylinder was selected, with grids of 100 equidistant points along the x-axis and 50 equidistant points along the y-axis, within the spatial domain \cite{lai2024temporal} of $[1,8] \times[-2,2]$. The initial data were gathered on a grid of 5000 points during the specified duration from 0 to 19.9, with an interval of 0.1. To create the original training set, we extracted 1,500 velocity data Figure~\ref{fig:training_data_distribution} points from a total of 1 million, selected randomly from time slices ranging from 0 to 19.9 at 0.1 intervals for the supervised objectives of flow reconstruction and inverse problems involving parameter identification. The mixture of the processed data points and equation points constituted the training set. The training set, comprising solely velocity information, was selected because the Navier-Stokes equations are characterized by velocity data. The Navier-Stokes constraints are included into the loss function to direct the model in precisely forecasting pressure gradients, ensuring that the predicted pressure progressively aligns with the values in the original data. The original dataset served as the validation set, encompassing both velocity information and actual pressure data.

\begin{algorithm}[H]
\caption{PhysicsFormer: An Efficient and Faster Transformer-Based Physics-Informed Neural Network for PDEs}
\label{alg:physicsformer}
\begin{algorithmic}[1]

\REQUIRE Training data $\mathcal{D} = \{(\mathbf{x}_i,\mathrm{t}_i), y_i\}$, differential operator $\mathcal{D}[\cdot]$, boundary operator $\mathcal{B}[\cdot]$, initial operator $\mathcal{I}[\cdot]$
\REQUIRE Hyperparameters: learning rates, loss weights $\lambda_{\{\text{residual, bc, ic, data}\}}$, tolerance, number of epochs $\mathbf{N}_{epochs}$

\STATE \textbf{Step 1: Architecture Initialization}
\STATE Initialize the Transformer encoder–decoder architecture $\mathcal{G}_\theta$ with Multi-Head Cross Attention in both encoder and decoder modules.
\STATE Employ a weighted sine activation function: 
\[
\phi(\mathrm{t}) = w \cdot \sin(\mathrm{t}), \quad w \in \mathbb{R}, \ \text{trainable parameter}
\]

\STATE \textbf{Step 2: Sequential Learning Setup with parallel model}
\STATE Encoder processes a sequential input of $k$ time steps $\{\mathrm{t}, \mathrm{t}+\Delta \mathrm{t}, \ldots, \mathrm{t}+(k-1)\Delta \mathrm{t}\}$.
\STATE Decoder outputs the corresponding $k$-step sequence: $\{\hat{\boldsymbol{u}}(\mathbf{x},\mathrm{t}), \ldots, \hat{\boldsymbol{u}}(\mathbf{x},\mathrm{t}+(k-1)\Delta \mathrm{t})\}$.

\STATE \textbf{Step 3: Physics-Guided Loss Formulation}
\STATE Compute PDE derivatives of $\hat{\boldsymbol{u}}$ with respect to $(\mathbf{x},\mathrm{t})$ using automatic differentiation.
\STATE Define the composite loss terms:
\begin{align*}
\mathcal{L}_{\text{residual}} &= \frac{1}{k \mathcal{N}_{\text{res}}} \sum_{i=1}^{\mathcal{N}_{\text{res}}} \sum_{\gamma=0}^{k-1}
\Big|\mathcal{D}\big[\hat{\boldsymbol{u}}(\mathbf{x}_i,\mathrm{t}_i+\gamma \Delta \mathrm{t})\big] - f(\mathbf{x}_i,\mathrm{t}_i+\gamma \Delta \mathrm{t})\Big|^2, \\
\mathcal{L}_{\text{bc}} &= \frac{1}{k \mathcal{N}_{\text{bc}}} \sum_{i=1}^{\mathcal{N}_{\text{bc}}} \sum_{\gamma=0}^{k-1}
\Big|\mathcal{B}\big[\hat{\boldsymbol{u}}(\mathbf{x}_i,\mathrm{t}_i+\gamma \Delta \mathrm{t})\big] - g(\mathbf{x}_i,\mathrm{t}_i+\gamma \Delta \mathrm{t})\Big|^2, \\
\mathcal{L}_{\text{ic}} &= \frac{1}{\mathcal{N}_{\text{ic}}} \sum_{i=1}^{\mathcal{N}_{\text{ic}}}
\Big|\mathcal{I}\big[\hat{\boldsymbol{u}}(\mathbf{x}_i,0)\big] - h(\mathbf{x}_i,0)\Big|^2, \\
\mathcal{L}_{\text{data}} &= \frac{1}{\mathcal{N}_d} \sum_{i=1}^{\mathcal{N}_d}
\Big|\hat{\boldsymbol{u}}(\mathbf{x}_d^i,\mathrm{t}_d^i) - y_d^i\Big|^2.
\end{align*}

\STATE Define total PhysicsFormer loss:
\[
\mathcal{L}_{\text{PhysicsFormer}} =
\lambda_{\text{residual}} \mathcal{L}_{\text{residual}}
+ \lambda_{\text{bc}} \mathcal{L}_{\text{bc}}
+ \lambda_{\text{ic}} \mathcal{L}_{\text{ic}}
+ \lambda_{\text{data}} \mathcal{L}_{\text{data}}
\]

\STATE \textbf{Step 4: Training Procedure}
\FOR{epoch $= 1$ to $\mathbf{N}_{epochs}$}
    \STATE Train the network using Adam optimizer.
    \IF{loss $\leq$ tolerance}
        \STATE \textbf{break}
    \ENDIF
\ENDFOR
\STATE Fine-tune the network parameters using L-BFGS optimizer until convergence.

\STATE \textbf{Step 5: Output}
\STATE Return optimized parameters $\theta^\ast$ and trained model $\hat{\boldsymbol{u}}(\mathbf{x},\mathrm{t};\theta^\ast)$.

\end{algorithmic}
\end{algorithm}

\refstepcounter{figure}

\begin{center}
    \includegraphics[width=0.6\textwidth]{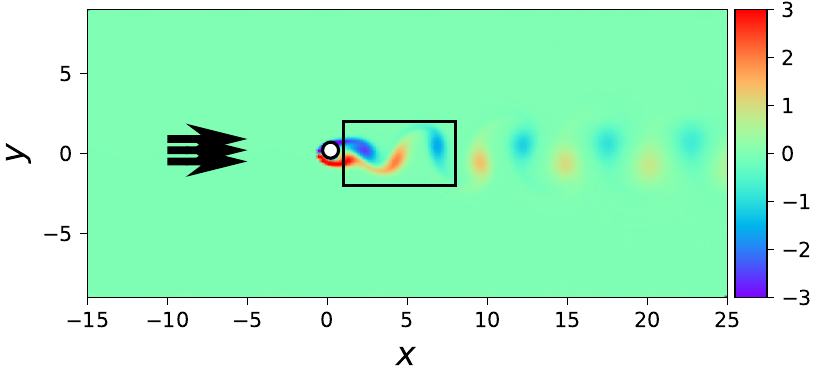}
\end{center}

\begin{quote}
\noindent\textbf{Figure \thefigure.}
The dark-shaded region indicates the vortex-shedding wake region for the flow past a circular cylinder considered in the present study.
\end{quote}

\label{fig:physicsformer_supervised}
\refstepcounter{figure}

\begin{center}
    \includegraphics[width=0.9\textwidth]{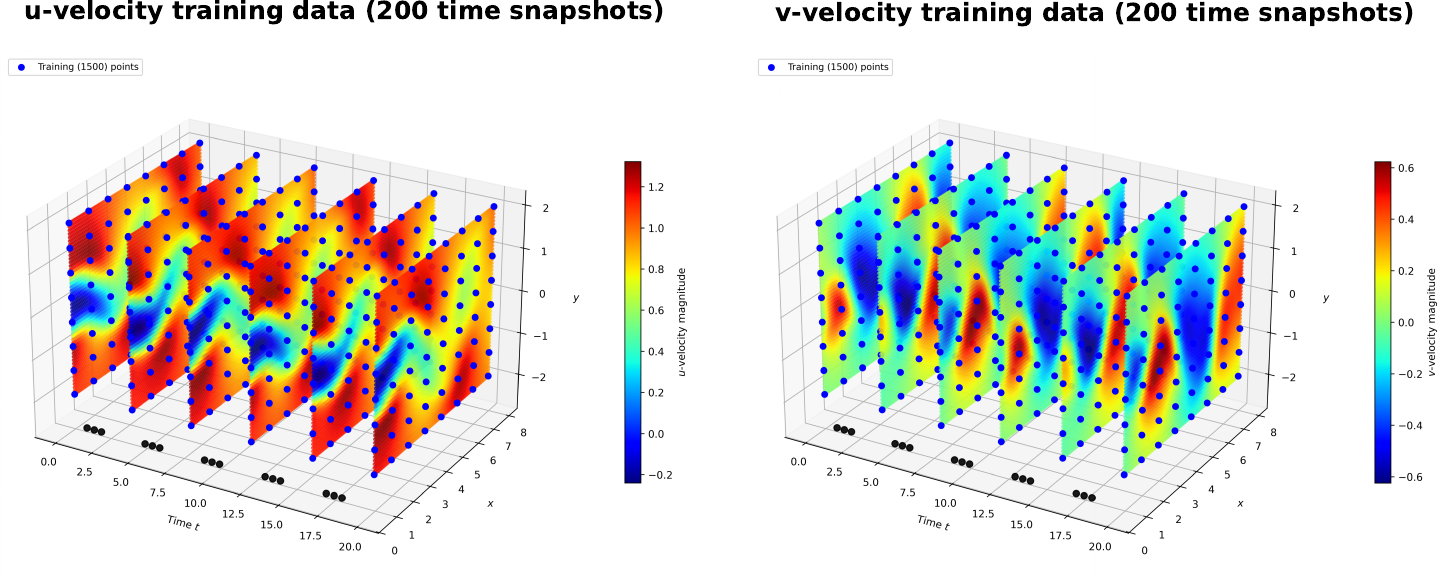}
\end{center}

\begin{quote}
\noindent\textbf{Figure \thefigure.}
Training data distribution for flow past a circular cylinder at $Re=100$. The left and right panels show the sampled $u$-velocity and $v$-velocity fields, respectively. A total of $1500$ spatio-temporal data points are used for training \textit{PhysicsFormer}.
\end{quote}

\label{fig:training_data_distribution}

\section{Numerical Study}
\label{sec:numerical_study}
The primary objective of this work is to develop a fast and computationally efficient Transformer-based physics-informed neural network framework capable of addressing the limitations and failure modes commonly observed in conventional PINNs \cite{krishnapriyan2021characterizing}. Motivated by this goal, we propose \textit{PhysicsFormer}, a lightweight attention-based framework designed for complex fluid flow applications governed by partial differential equations. To evaluate the effectiveness of the proposed method, several benchmark problems are investigated, including the forward solution of the one-dimensional Burgers' equation containing shock structures, inverse parameter identification for the incompressible Navier--Stokes equations involving the unknown convection coefficient ($\lambda_1$) and diffusion coefficient ($\lambda_2$), and flow reconstruction problems governed by the two-dimensional incompressible Navier--Stokes equations. In addition, we investigate lid-driven cavity flow at $Re=100$ and flow reconstruction problems at both $Re=100$ and the high-Reynolds-number regime of $Re=3900$. For the cylinder wake reconstruction problem, the proposed framework accurately reconstructs the velocity, pressure, vorticity, and streamline fields using sparse observational data. Compared with PINNsFormer, the proposed framework achieves improved computational efficiency while requiring fewer training samples. In the Navier--Stokes reconstruction problem, only $1,500$ data points, corresponding to approximately $0.15\%$ of the total dataset containing $1,000,000$ points, are utilized for training, whereas PINNsFormer requires $2,500$ data points. Even when trained using the same $2,500$ data samples, \textit{PhysicsFormer} demonstrates superior predictive performance. The predicted pressure fields and absolute error distributions show clear qualitative improvements, as illustrated in Figure~\ref{fig:comparison_pdf}. For the Burgers' equation, the collocation points are selected as $\mathcal{N}_{res}=2601$, while the boundary and initial condition points are chosen as $\mathcal{N}_{bc}=\mathcal{N}_{ic}=51$. The obtained results are compared with several existing approaches, including PINNs \cite{raissi2019physics}, PINNsFormer \cite{zhao2023pinnsformer, zhaopinnsformer}, QRes \cite{bu2021quadratic}, and First-Layer Sine (FLS) networks \cite{wong2022learning}. Both qualitative visualizations and quantitative error analyses demonstrate the strong predictive capability and computational efficiency of the proposed \textit{PhysicsFormer} framework.

\subsection{Mitigate Failure Modes of PINNs}
\textbf{Convection PDE.} The one-dimensional convection problem is a hyperbolic partial differential equation frequently employed to represent transport phenomena. The system is defined by the formulation incorporating periodic boundary conditions as follows:

$$
\begin{gathered}
u_t + \beta u_x = 0, 
\quad \forall\, x \in [0, 2\pi],\; t \in [0, 1], \\
\text{Initial Condition: } u(x, 0) = \sin(x), 
\quad \text{Boundary Condition: } u(0, t) = u(2\pi, t).
\end{gathered}
$$

where $\beta$ is the convection coefficient. As $\beta$ grows, the solution's frequency increases, making it more challenging for PINNs to estimate. In this instance, we establish $\beta=50$.  This is a significant failure mode of PINNs, as they are unable to capture high-frequency solutions. In contrast, our proposed \textit{PhysicsFormer} accurately predicts higher-frequency solutions by incorporating temporal dependencies into the learning process. We implement soft regularization for PINNs to address that problem and optimize the loss function.  Subsequent to training, we evaluate the realative errors between the predicted solution of the PINN and the analytical solution \cite{zhao2023pinnsformer, zhaopinnsformer}, as illustrated in Figure~\ref{fig:convection_beta50}. The PINN demonstrates efficacy in obtaining satisfactory solutions only for low convection coefficients, failing to perform efficiently as $\beta$ increases, resulting in a relative error approaching $100\%$ at $\beta=50$. Our proposed \textit{PhysicsFormer} effectively captures high-frequency solutions, yielding exact results with absolute error of approximately $1 \times 10^{-5}$.

\refstepcounter{figure}

\begin{center}
    \includegraphics[width=0.80\textwidth]{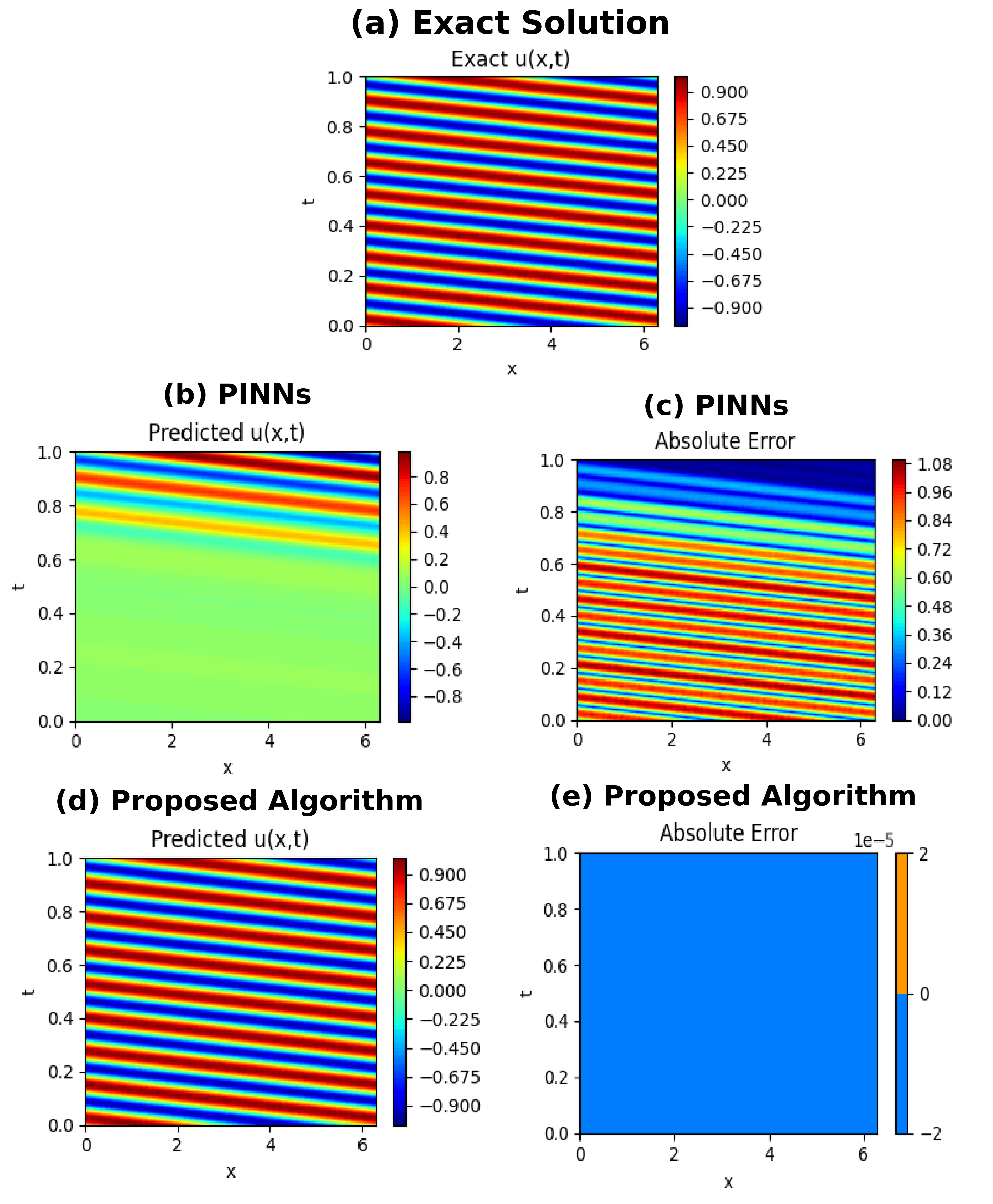}
\end{center}

\noindent\textbf{Figure \thefigure.}
Comparison of the convection equation results at $\beta=50$ obtained using PINNs and the proposed \textit{PhysicsFormer}. Panel (a) presents the exact solution reported in~\cite{zhao2023pinnsformer, zhaopinnsformer}. Panels (b) and (d) correspond to the prediction results generated by PINNs and \textit{PhysicsFormer}, respectively, while panels (c) and (e) show the associated absolute error distributions.

\label{fig:convection_beta50}

\subsection{Forward Problem (Burger's equation):}
In one spatial dimension, the Burgers' equation with Dirichlet boundary conditions is expressed as

\begin{equation}
\begin{aligned}
u_t + u\,u_x - \frac{0.01}{\pi} u_{xx} &= 0, 
\quad x \in [-1,1], \; t \in [0,1], \\
u(0,x) &= -\sin(\pi x), \\
u(t,-1) &= u(t,1) = 0 .
\end{aligned}
\label{eq:burgers}
\end{equation}
where $\nu=\frac{0.01}{\pi}>0$ is the coefficient of kinematic viscosity.  Equation~\ref{eq:burgers} was initially presented by Bateman \cite{bateman1915some} and then investigated by Burgers \cite{burgers1948mathematical}, after whom the equation is commonly known as Burgers' equation. This equation is essential in the investigation of nonlinear waves, offering as a mathematical model in turbulence problems, shock wave theory, and continuous stochastic processes. A wide range of scientists is dedicated to examining the exact and numerical solutions of this equation.

This study introduces a novel hybrid architecture, \textit{PhysicsFormer}, that combines the capabilities of the Transformer framework with the physics-informed neural network paradigm. The model is designed to capture global dependencies in the input data via cross-attention mechanisms and to enforce local physics-based constraints through the PDE residuals. The input undergoes a linear transformation into a higher-dimensional embedding space, which is then processed using the encoder-decoder architecture of the Transformer. Each block employs multi-head cross-attention layers to capture long-range dependencies, while residual connections and weighted sine-based activation functions ($\phi(\mathrm{t})$) promote stable gradient flow and improve the representation of oscillatory solutions. The output is then reintroduced into physical space through a concluding feed-forward network.

The \textit{PhysicsFormer} framework uses a physics-informed loss function for training and integrates multiple input sources. During the training of \textit{PhysicsFormer}, we utilize collocation points, with $\mathcal{N}_{ic} = \mathcal{N}_{bc} = 51$ initial and boundary points, and a $51 \times 51$ grid, resulting in $\mathcal{N}_{res} = 2601$ residual points (grid points) Figure~\ref{fig:burgers_training_combined}. This ensures that the model complies with essential physical principles while also satisfying the stated initial and boundary constraints. The collocation method allows the model to generalize efficiently throughout the domain without requiring large labeled data, which is sometimes difficult or expensive to obtain in real-world fluid dynamics situations.

The model is configured using hyperparameters carefully chosen to enhance accuracy and computational efficiency. The embedding dimension was set to $d_{\text{model}}=32$, the hidden dimension of the output MLP to $d_{\text{hidden}}=512$, the number of encoder/decoder layers to $N=1$, the number of attention heads to $2$, and the output dimension to $1$ to represent the scalar field $u$. We employ Xavier uniform initialization for weight initialization, with a small bias of $0.01$ to improve stability during the first training phase.   The optimization employs the L-BFGS algorithm alongside a robust Wolfe line search, which is especially efficient for the smooth loss landscapes of PINNs. Our proposed framework employs approximately $500$ MB of GPU memory, making it lightweight and suitable for deployment on any modern GPU card.   The model was trained on a Google Colab T4 GPU, with an overall training duration of approximately 20 minutes.   Our framework is around twice times more efficient than PINNsFormer, mostly because to the effectiveness of the weighted \textit{Sine} activation function.   We trained our proposed model for $500$ epochs to achieve convergence, resulting in a total loss of $6.0 \times 10^{-6}$, while the physics residual loss decreased to $5.0 \times 10^{-7}$. In the hyperparameter specifications of the PhysicsFormer Burger's equation, we presented Table \ref{tab:physicformer_burger}, indicating that the relative error $L_2$ for these configurations is $2.4 \times 10^{-4}$, compared to $6.7 \times 10^{-4}$ seen in PINNs.

For qualitative comparison, we presented the exact solution of Burger's equation alongside our proposed \textit{PhysicsFormer} prediction and the corresponding absolute error in Figure~\ref{fig:burgers}. Furthermore, Figure~\ref{fig:comparison} illustrates the solution to Burger's equation after predicting three separate time intervals: $t=0.25s$, $t=0.50s$, and $t=0.75s$, indicating that our anticipated answer roughly corresponds with the exact solution. The graph in Figure~\ref{fig:burgers_training_combined} demonstrates that the training loss for the forward issue of Burger's equation attains steady convergence after 500 epochs.

\begin{table}[htbp]
\centering
\caption{\textit{PhysicFormer} model architecture for Burger's equation}
\label{tab:physicformer_burger}
\begin{tabular}{ll}
\hline
\textbf{Parameter} & \textbf{Value} \\
\hline
Transformer embedding dimension ($d_{model}$) & 32 \\
Hidden layer size in output MLP ($d_{hidden}$) & 512 \\
Number of encoder/decoder layers ($N$) & 1 \\
Number of attention heads & 2 \\
Output dimension ($d_{out}$) & 1 \\
Optimizer & L-BFGS (line search: strong Wolfe) \\
Weight initialization & Xavier Uniform, bias = 0.01 \\
Total epochs & 500 \\
\hline
\end{tabular}
\end{table}

\refstepcounter{figure}

\begin{center}
    \includegraphics[width=1.0\linewidth]{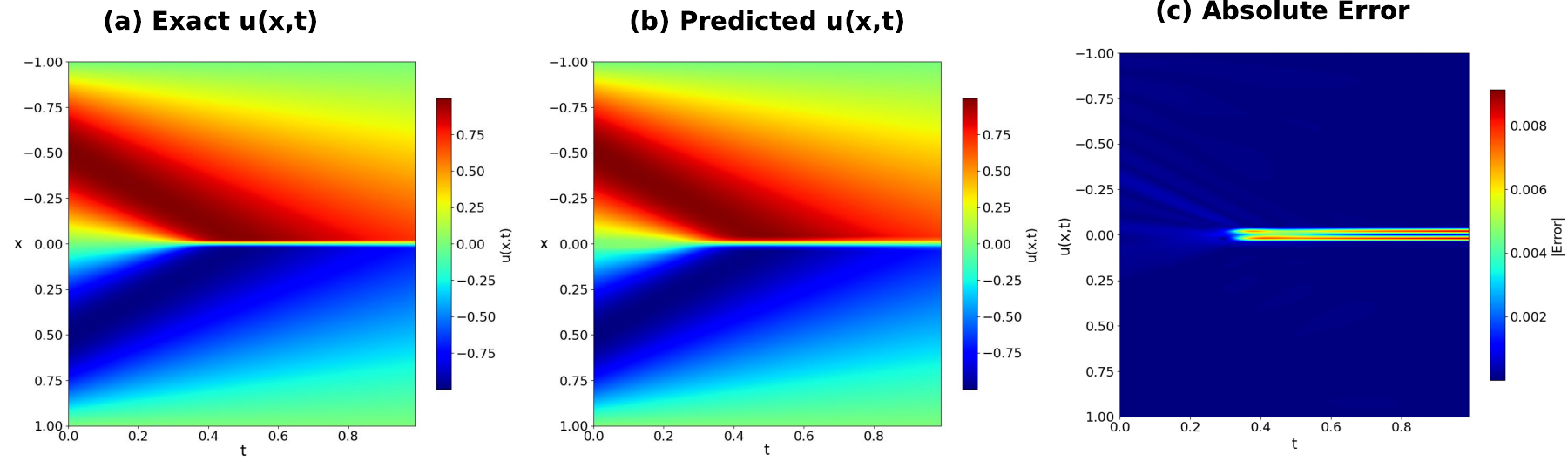}
\end{center}

\noindent\textbf{Figure \thefigure.}
Comparison of Burgers' equation solutions: (a) exact reference solution reported in~\cite{raissi2019physics}, (b) prediction obtained using the proposed \textit{PhysicsFormer} framework, and (c) corresponding absolute error distribution of the proposed method.

\label{fig:burgers}

\refstepcounter{figure}

\begin{center}
    \includegraphics[width=1.0\linewidth]{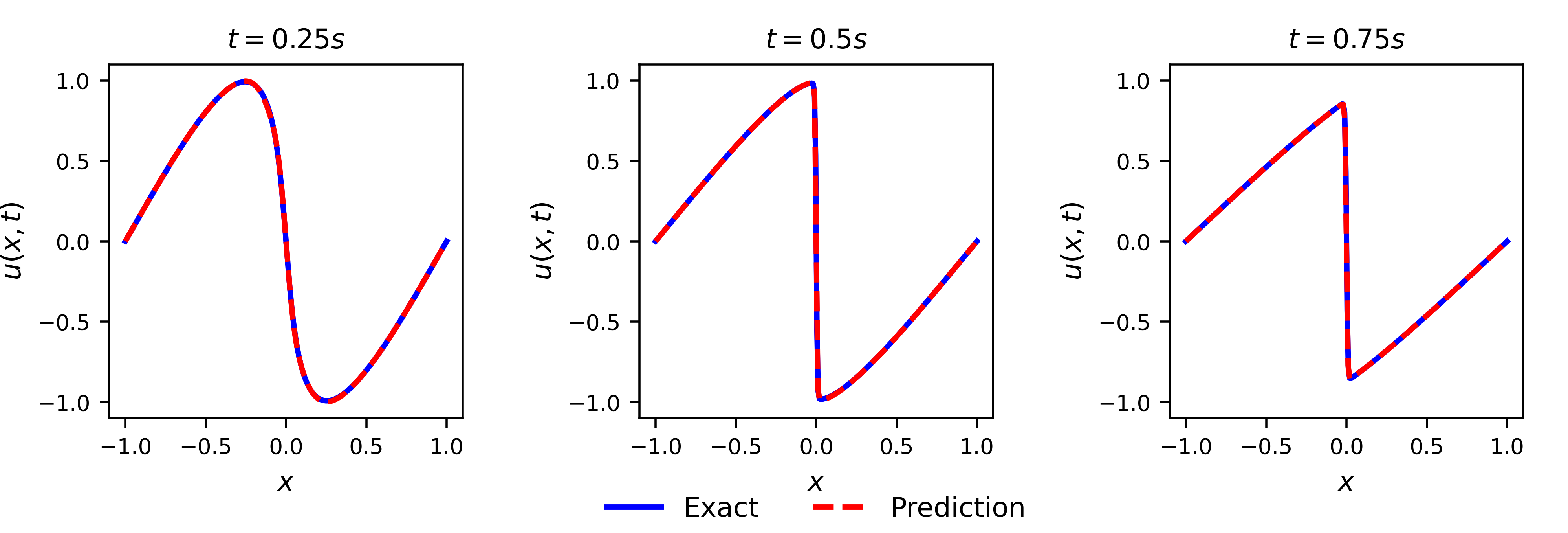}
\end{center}

\noindent\textbf{Figure \thefigure.}
Comparison between the predicted and exact solutions of Burgers' equation at three temporal instances: $t=0.25\,\mathrm{s}$, $t=0.50\,\mathrm{s}$, and $t=0.75\,\mathrm{s}$. The proposed \textit{PhysicsFormer} accurately captures the evolution of the shock structure, with the predicted profiles remaining in close agreement with the exact solutions across all time snapshots.

\label{fig:comparison}

\refstepcounter{figure}

\begin{center}
    \includegraphics[width=1.0\textwidth]{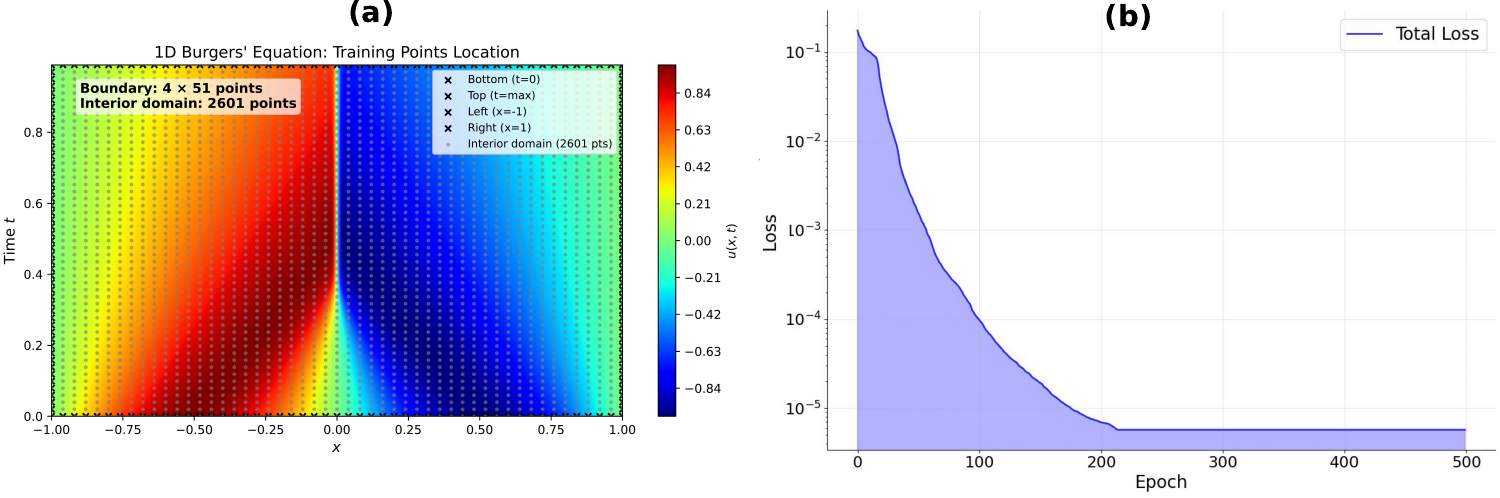}
\end{center}

\noindent\textbf{Figure \thefigure.}
Training data configuration and convergence behavior for Burgers' equation. Panels (a) and (b) illustrate the distribution of collocation/training points and the evolution of the training loss with respect to epochs, respectively.

\label{fig:burgers_training_combined}

\subsection{Lid-Driven Cavity Flow ($Re = 100$)}

To assess the effectiveness of the proposed framework, we consider the classical lid-driven cavity flow problem, which serves as a standard benchmark in computational fluid dynamics. This problem is governed by the steady incompressible Navier--Stokes equations and is widely used to evaluate numerical methods due to its well-defined boundary conditions and nonlinear flow characteristics.

The governing equations in two dimensions are given by
\begin{align}
\frac{\partial u}{\partial x} + \frac{\partial v}{\partial y} &= 0, \label{eq:continuity} \\
u \frac{\partial u}{\partial x} + v \frac{\partial u}{\partial y} + \frac{1}{\rho} \frac{\partial p}{\partial x} - \nu \left( \frac{\partial^2 u}{\partial x^2} + \frac{\partial^2 u}{\partial y^2} \right) &= 0, \label{eq:momentum_x} \\
u \frac{\partial v}{\partial x} + v \frac{\partial v}{\partial y} + \frac{1}{\rho} \frac{\partial p}{\partial y} - \nu \left( \frac{\partial^2 v}{\partial x^2} + \frac{\partial^2 v}{\partial y^2} \right) &= 0, \label{eq:momentum_y}
\end{align}

where $(u, v)$ denote the velocity components in the $x$ and $y$ directions, respectively, and $p$ represents the pressure field. The fluid density is taken as $\rho = 1$, and the kinematic viscosity is set to $\nu = 0.01$, corresponding to a Reynolds number $Re = 100$.

The computational domain is defined as a unit square $(x, y) \in [0,1] \times [0,1]$. The boundary conditions are prescribed as follows:
\begin{itemize}
\item \textbf{Top boundary ($y = 1$):} a constant horizontal velocity is imposed, such that $u = 1$ and $v = 0$.
\item \textbf{Remaining boundaries ($y = 0$, $x = 0$, $x = 1$):} no-slip conditions are enforced, i.e., $u = 0$ and $v = 0$.
\end{itemize}

This configuration generates a primary vortex within the cavity along with secondary flow structures near the corners, making it a suitable test case for evaluating the accuracy and stability of the proposed model.

To solve the lid-driven cavity problem using our proposed framework, PhysicsFormer, we represent the incompressible velocity field through a stream function $\psi$. The velocity components are obtained automatically by automatic differentiation as
\begin{align}
u &= \frac{\partial \psi}{\partial y}, \\
v &= -\frac{\partial \psi}{\partial x}.
\end{align}
This formulation satisfies the continuity equation,
\begin{equation}
\frac{\partial u}{\partial x} + \frac{\partial v}{\partial y} = 0,
\end{equation}
exactly for any differentiable $\psi$. As a result, the learning task is reduced to enforcing only the momentum equations and the boundary constraints.

The steady Navier--Stokes residuals for the two-dimensional cavity flow are written as
\begin{align}
f_u &= u\frac{\partial u}{\partial x} + v\frac{\partial u}{\partial y} + \frac{\partial p}{\partial x}
- \frac{1}{Re}\left(\frac{\partial^2 u}{\partial x^2} + \frac{\partial^2 u}{\partial y^2}\right), \\
f_v &= u\frac{\partial v}{\partial x} + v\frac{\partial v}{\partial y} + \frac{\partial p}{\partial y}
- \frac{1}{Re}\left(\frac{\partial^2 v}{\partial x^2} + \frac{\partial^2 v}{\partial y^2}\right).
\end{align}

The interior physics loss is defined at the collocation points as
\begin{equation}
\mathcal{L}_{\mathrm{PDE}} =
\frac{1}{\mathcal{N}_{\mathrm{int}}}\sum_{i=1}^{\mathcal{N}_{\mathrm{int}}}
\left(
\|f_u^{(i)}\|^2 + \|f_v^{(i)}\|^2
\right).
\end{equation}

The boundary conditions are enforced on all walls. No-slip conditions are applied on the stationary walls, while the top lid moves with unit horizontal velocity:
\begin{equation}
u=1,\quad v=0 \qquad \text{on } y=1,
\end{equation}
and
\begin{equation}
u=0,\quad v=0 \qquad \text{on } x=0,\; x=1,\; y=0.
\end{equation}
Accordingly, the velocity boundary losses are expressed as
\begin{align}
\mathcal{L}_{\mathrm{bc},u} &= \frac{1}{\mathcal{N}_{\mathrm{bnd}}}\sum_{j=1}^{\mathcal{N}_{\mathrm{bnd}}}
\left(u_{\mathrm{pred}}^{(j)} - u_{\mathrm{true}}^{(j)}\right)^2, \\
\mathcal{L}_{\mathrm{bc},v} &= \frac{1}{\mathcal{N}_{\mathrm{bnd}}}\sum_{j=1}^{\mathcal{N}_{\mathrm{bnd}}}
\left(v_{\mathrm{pred}}^{(j)} - v_{\mathrm{true}}^{(j)}\right)^2.
\end{align}

Since the stream function must vanish on the cavity walls, we also impose
\begin{equation}
\psi = 0 \qquad \text{on all boundaries},
\end{equation}
with the corresponding penalty term
\begin{equation}
\mathcal{L}_{\mathrm{bc},\psi} =
\frac{1}{\mathcal{N}_{\mathrm{bnd}}}\sum_{j=1}^{\mathcal{N}_{\mathrm{bnd}}}
\left(\psi_{\mathrm{pred}}^{(j)}\right)^2.
\end{equation}

The total training objective is then given by
\begin{equation}
\mathcal{L}
=
\mathcal{L}_{\mathrm{PDE}}
+
\mathcal{L}_{\mathrm{bc},u}
+
\mathcal{L}_{\mathrm{bc},v}
+
\mathcal{L}_{\mathrm{bc},\psi}.
\end{equation}

For this experiment, the PhysicsFormer architecture uses a sequence length $k=5$, a pseudo-time step $\Delta t = 10^{-3}$, $d_{\mathrm{model}}=32$, $d_{\mathrm{ff}}=256$, $d_{\mathrm{hidden}}=512$, two attention heads, one encoder layer, one decoder layer, and learnable sine activation functions. The model is trained for 2000 epochs using the L-BFGS optimizer with strong-Wolfe line search.

After training, the final total loss reaches approximately $8\times10^{-5}$. The rRMSE errors are $5\times10^{-3}$ for $u$, $7\times10^{-3}$ for $v$, and $6\times10^{-2}$ for pressure. These results indicate that PhysicsFormer achieves notably better accuracy than the low-order LA-PINN~\cite{song2024loss} baseline. The quantitative comparison is reported in Table~\ref{tab:comparison_rRMSE}, while the predicted fields of $u$, $v$, and pressure are visualized in Fig.~\ref{fig:lid_cavity_results}. In addition, the validation against the benchmark solution of Ghia et al.~\cite{ghia1982high} is shown in Fig.~\ref{fig:ghia_validation}.

\begin{table}[htbp]
\centering
\caption{Comparison of relative RMSE(rRMSE) for velocity components and pressure between LA-PINN~\cite{song2024loss} and PhysicsFormer for the lid-driven cavity flow at $Re=100$.}
\label{tab:comparison_rRMSE}
\begin{tabular}{lccc}
\hline
\textbf{Model} & $\mathbf{u}$-velocity & $\mathbf{v}$-velocity & \textbf{Pressure} \\
\hline
LA-PINN~\cite{song2024loss}       & $3.84 \times 10^{-2}$ & $4.92 \times 10^{-2}$ & $2.04 \times 10^{-1}$ \\
\textbf{PhysicsFormer}  & $\mathbf{5 \times 10^{-3}}$ & $\mathbf{7 \times 10^{-3}}$ & $\mathbf{6 \times 10^{-2}}$ \\
\hline
\end{tabular}
\end{table}

\refstepcounter{figure}

\begin{center}
    \includegraphics[width=\textwidth]{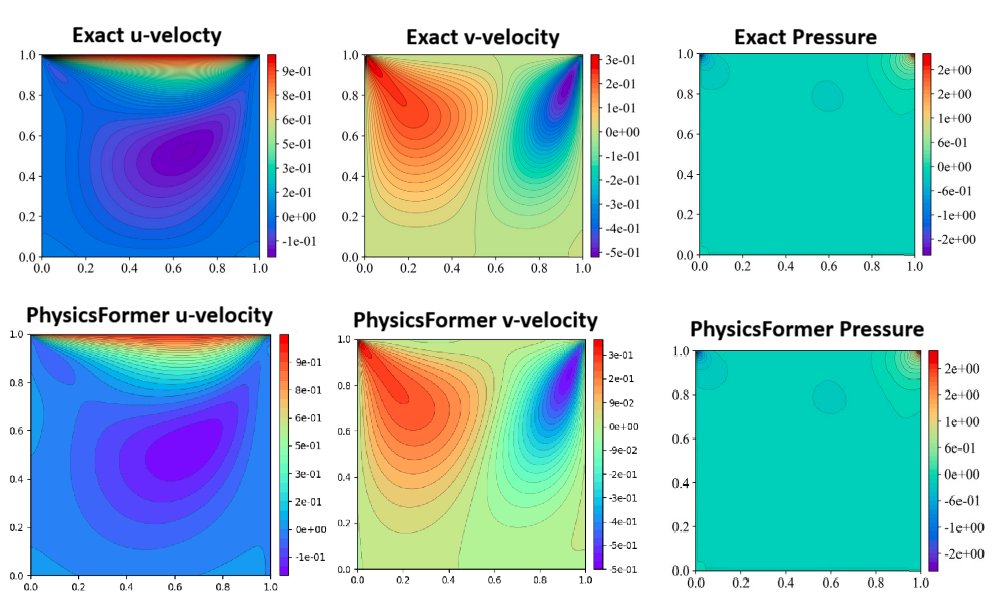}
\end{center}

\vspace{0.5em}

\noindent\textbf{Figure \thefigure.}
Comparison of flow fields for the lid-driven cavity problem at $Re=100$. The figure presents the reference (exact) solution, PhysicsFormer predictions, and corresponding absolute error distributions for the $u$-velocity, $v$-velocity, and pressure fields.

\label{fig:lid_cavity_results}

\refstepcounter{figure}

\begin{center}
    \includegraphics[width=0.90\textwidth]{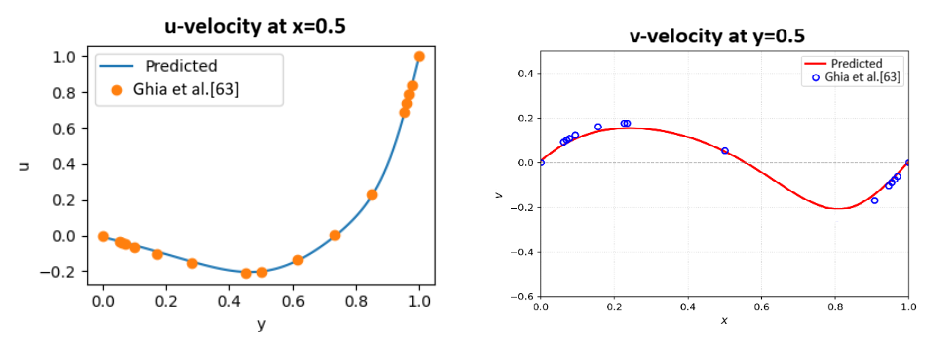}
\end{center}

\vspace{0.5em}

\noindent\textbf{Figure \thefigure.}
Validation of the lid-driven cavity flow solution at $Re=100$ against the benchmark data of Ghia et al.~\cite{ghia1982high} The left panel shows the horizontal velocity ($u$) profile along the vertical centerline ($x=0.5$), while the right panel presents the vertical velocity ($v$) profile along the horizontal centerline ($y=0.5$). The results obtained using PhysicsFormer are in close agreement with the reference data.

\label{fig:ghia_validation}

\subsection{Flow Field Reconstruction of Cylinder Wake at $Re=100$}
To illustrate the model's enhancement, a fundamental problem involving the flow of an incompressible fluid is selected, characterized by the Navier-Stokes (NS) equations. The Navier-Stokes equations describe the dynamics of viscous, incompressible fluids. They are essential in fluid dynamics, representing the conservation of momentum and mass. It can model and anticipate diverse fluid flow issues, encompassing the movement of gases and liquids, turbulent processes, and the efficacy of wind turbines. In recent years, various numerical methods have been utilized to solve the Navier-Stokes equations, yielding quantitative results for fluid flow, which serve as a foundation for optimizing designs and informing engineering decisions, including aircraft design, hydraulic engineering planning, and weather forecasting.

In this study we examine the reconstruction of incompressible fluid flows determined by the two-dimensional Navier-Stokes equations. The equation are expressed in velocity-pressure form as
\begin{equation}
\label{eq:NS}
\begin{aligned}
u_t + u u_x + v u_y &= -p_x + \frac{1}{Re}(u_{xx}+u_{yy}),\\
v_t + u v_x + v v_y &= -p_y + \frac{1}{Re}(v_{xx}+v_{yy}),\\
u_x + v_y &= 0,
\end{aligned}
\end{equation}

Here, $u(x,y,t)$ and $v(x,y,t)$ denote the velocity components, $p(x,y,t)$ represents the pressure, and $Re=100$ signifies the Reynolds number. To ensure incompressibility, we utilize a stream function $\psi(x,y,t)$ such that
$$
u = \frac{\partial \psi}{\partial y}, \quad v = -\frac{\partial \psi}{\partial x},
$$
It guarantees that $\nabla \cdot \mathbf{u} = 0$ holds universally.  The vorticity field is subsequently derived for analysis as
$$
\omega = v_x - u_y.
$$

The learning problem is defined as reconstructing the fields $\{u,v,p,\psi,\omega\}$ from significantly sparse observations.  In practice, merely $1{,}500$ labeled velocity samples, representing $0.15\%$ of a reference dataset comprising $10^6$ points, are utilized for training, although pressure is never directly measured. This illustrates true situations in which obtaining dense fluid data is prohibitively costly. The limited supervision is enhanced by physical restrictions, guaranteeing that the reconstructed solutions align with the governing equations.

We propose a transformer-inspired physics-informed network, designated as \textit{PhysicsFormer}, to accomplish this objective. The model transforms the spatio-temporal inputs $(x,y,t)$ into a latent representation of dimension $d_{\text{model}}=32$, which is then processed by $N=1$ stacked encoder-decoder layer including $4$ attention heads. Each block comprises multi-head self-attention succeeded by feed-forward layers utilizing weighted \textit{Sine} activation functions.

$$
\phi(\mathrm{t}) = w \sin(\mathrm{t}),
$$
where $w$ is a parameter subject to training.  This selection improves the expressiveness of oscillatory dynamics characteristic of fluid flows.  The hidden layer dimension is configured to $128$, and the output head generates two channels $(\psi,p)$.  The velocity components are subsequently obtained through automatic differentiation of $\psi$.

This architecture guarantees that incompressibility is precisely fulfilled via the stream function formulation. Additionally, vorticity and pressure are simultaneously reconstructed with the velocity field, providing an accurate representation of the flow.  The total number of trainable parameters is carefully adjusted to ensure effective optimization and accuracy.

The model parameters are initialized using an appropriate weight initialization technique and trained with the Limited-memory Broyden–Fletcher–Goldfarb–Shanno (L-BFGS) \cite{liu1989limited} optimizer employing a strong Wolfe line search.  Throughout the training process, at each iteration, the model generates $(\psi,p)$, from which the velocity components $(u,v)$ are automatically derived.  Higher-order derivatives concerning $(x,y,t)$ are derived using automatic differentiation.

The residuals of the governing PDEs are then computed as
\begin{equation}
\label{eq:residuals}
\begin{aligned}
f_u &= u_t + (u u_x + v u_y) + p_x - \frac{1}{Re} (u_{xx}+u_{yy}), \\
f_v &= v_t + (u v_x + v v_y) + p_y - \frac{1}{Re} (v_{xx}+v_{yy}),
\end{aligned}
\end{equation}
which correspond to the $x$- and $y$-momentum equations. The total loss is defined as
\begin{equation}
\mathcal{L_{\text{PhysicsFormer}}} = 
\frac{1}{\mathcal{N}_u} \sum_{i=1}^{\mathcal{N}_u} 
\left( \hat{u}_i - u_i \right)^2 
+ \frac{1}{\mathcal{N}_v} \sum_{i=1}^{\mathcal{N}_v} 
\left( \hat{v}_i - v_i \right)^2 
+ \frac{1}{\mathcal{N}_f} \sum_{i=1}^{\mathcal{N}_f} 
\left( f_u(\mathbf{x}_i, t_i) \right)^2 
+ \frac{1}{\mathcal{N}_f} \sum_{i=1}^{\mathcal{N}_f} 
\left( f_v(\mathbf{x}_i, t_i) \right)^2 
\end{equation}

The initial two summations address the deviation between the estimated velocity components $(\hat{u}_i, \hat{v}_i)$ and their respective reference values $(u_i, v_i)$, namely data loss, normalized by the quantity of data points $\mathcal{N}_u$ and $\mathcal{N}_v$. The final two summations ensure compliance with the governing equations by imposing penalties on the residuals $f_u(\mathbf{x}_i, t_i)$ and $f_v(\mathbf{x}_i, t_i)$, which are averaged over $\mathcal{N}_f$ collocation points in both space and time.  The concept reconciles conformity with empirical evidence while concurrently integrating physical rules into the learning process. The vorticity $\omega$ is calculated as $\omega = v_x - u_y$, but it is excluded from the loss function and preserved for subsequent analysis and visualization in Figure~\ref{fig:cfd_vs_physicsformer_vorticity}.

The training persists for $1000$ epochs, and the model hyperparameters are presented in Table~\ref{tab:navierstokes_hyperparameter}. During each epoch, a closure function computes the loss, performs backpropagation of gradients, and modifies the model parameters. This joint optimization ensures precise reconstruction of sparse data while preserving dynamic consistency throughout the whole flow field. Thus, the framework can accurately reconstruct velocity, pressure, streamline, and vorticity fields from severely limited data.

To ensure a fair comparison with PINNsFormer~\cite{zhao2023pinnsformer, zhaopinnsformer}, PINNs~\cite{raissi2019physics}, QRes~\cite{bu2021quadratic}, and FLS~\cite{wong2022learning}, all models were trained using identical hyperparameters and the same dataset containing $2500$ samples. All computations were performed on NVIDIA L4 $24$ GB and Google Colab T4 $15$ GB GPUs, as summarized in Table~\ref{tab:comp_time}. The proposed \textit{PhysicsFormer} achieves improved pressure reconstruction and lower absolute error compared with the existing methods, as illustrated in Fig.~\ref{fig:comparison_pdf}. In addition, Tables~\ref{tab:comp_time} and~\ref{tab:time_memory} demonstrate that \textit{PhysicsFormer} requires substantially lower GPU memory and nearly half the computational time of PINNsFormer~\cite{zhao2023pinnsformer, zhaopinnsformer}. The proposed framework also attains lower rMAE and rRMSE values for the pressure field, as reported in Table~\ref{tab:navierstokes_results}, indicating improved prediction accuracy over PINNs~\cite{raissi2019physics}, QRes~\cite{bu2021quadratic}, FLS~\cite{wong2022learning}, and PINNsFormer~\cite{zhao2023pinnsformer}. Furthermore, even when trained for only $2000$ epochs, \textit{PhysicsFormer} exhibits lower training loss and more stable convergence behavior than the compared methods, as shown in Fig.~\ref{fig:training_loss}, demonstrating the robustness and computational efficiency of the proposed framework.

To evaluate the precision of the proposed physics-informed neural network model, we utilize relative Mean Absolute Error (rMAE) and relative Root Mean Square Error (rRMSE) as assessment metrics, the comparison of which is presented in Table~\ref{tab:navierstokes_results}. The specific formulations are delineated as follows:

\begin{equation}
    \text{rMAE} = 
    \frac{\sum\limits_{n=1}^{\mathcal{N}} \left| \hat{\boldsymbol{u}}(x_n,\mathrm{t}_n) - \boldsymbol{u}(x_n,\mathrm{t}_n) \right|} 
         {\sum\limits_{n=1}^{\mathcal{N}_{\text{res}}} \left| \boldsymbol{u(}x_n,\mathrm{t}_n) \right|}
    \label{eq:rmae}
\end{equation}

\begin{equation}
\mathrm{rRMSE} \;=\;
\sqrt{
\frac{\displaystyle \sum_{n=1}^{\mathcal{N}} \bigl|\hat{\boldsymbol{u}}(x_n,\mathrm{t}_n)-\boldsymbol{u}(x_n,\mathrm{t}_n)\bigr|^{2}}
     {\displaystyle \sum_{n=1}^{\mathcal{N}} \bigl|\boldsymbol{u}(x_n,\mathrm{t}_n)\bigr|^{2}}
}\ 
\label{eq:rrmse_mod}
\end{equation}

where $\mathcal{N}$ represents the total number of testing points, $\hat{\boldsymbol{u}}$ denotes the neural network approximation, and $\boldsymbol{u}$ corresponds to the ground truth solution. Here, $\mathcal{N}_{\text{res}}$ indicates the number of residual collocation points employed in the training process.

To improve model efficacy, we train our proposed \textit{PhysicsFormer} with $0.15\%$, or $1500$, velocity, and our model accurately reconstructs the velocity, pressure, vorticity, and streamline in the wake region.   Figures~\ref{fig:uvp_vs_physicsformer}, \ref{fig:cfd_vs_physicsformer_vorticity}, and \ref{fig:streamline_comparison} illustrate a comparison between our findings and experimental data; our model correctly reconstructs the flow field. For training purposes Figure~\ref{fig:physicsformer_supervised}, we show velocity data from $t=0.00s$ to $t=19.90s$ with $0.1s$ time slices, while for validation, we reconstruct unseen velocity data and the pressure field at $t=20s$.

\begin{table}[htbp]
\centering
\caption{Schematic representation of the proposed \textit{PhysicsFormer} framework for flow reconstruction of the two-dimensional incompressible Navier--Stokes equations at $Re=100$.}
\begin{tabular}{ll}
\hline
Hyperparameter & Value \\
\hline
$d_{\text{out}}$ & 2 \\
$d_{\text{hidden}}$ & 128 \\
$d_{\text{model}}$ & 32 \\
$N$ (Layers) & 1 \\
Heads & 4 \\
$d_{ff}$ & 256 \\
Activation & Weighted Sine $\; \phi(t) = w \sin(t)$ \\
Epochs & 1000 \\
Optimizer & L-BFGS \\
\hline
\end{tabular}
\label{tab:navierstokes_hyperparameter}
\end{table}

\begin{table}[h!]
\centering
\caption{Comparison of flow reconstruction performance at $Re=100$ using Loss, rMAE, and rRMSE metrics.}
\begin{tabular}{lccc}
\hline
\textbf{Model} & \textbf{Loss} & \textbf{rMAE} & \textbf{rRMSE} \\
\hline
PINNs         & 6.72e$^{-5}$ & 13.08  & 9.08 \\
QRes           & 2.24e$^{-4}$ & 6.41   & 4.45 \\
FLS          & 9.54e$^{-6}$ & 3.98   & 2.77 \\
PINNsFormer \cite{zhao2023pinnsformer, zhaopinnsformer}   & 6.66e$^{-6}$ & 0.384 & 0.280 \\
\textbf{Proposed Agorithm} & \textbf{5.35e$^{-6}$} & \textbf{0.136} & \textbf{0.133} \\ 
\hline
\end{tabular}
\label{tab:navierstokes_results}
\end{table}

\begin{table}[htbp]
\centering
\caption{Comparison of computational cost and model complexity for flow reconstruction at $Re=100$ with fixed $k=5$ and $\Delta t = 1\times10^{-2}$.}
\label{tab:comp_time}
\begin{tabular}{lccc}
\hline
\textbf{Model} & \textbf{Total Computational Time} & \textbf{GPU Card Details} & \textbf{Model Parameters} \\
\hline
PINNsFormer \cite{zhao2023pinnsformer, zhaopinnsformer}   & 184 min & L4, 24GB & 454,106 \\
\textbf{Proposed Algorithm} & \textbf{60 min} & L4, 24GB & \textbf{194510} \\
\hline
\end{tabular}
\end{table}

\refstepcounter{figure}

\begin{center}
    \includegraphics[width=0.6\textwidth]{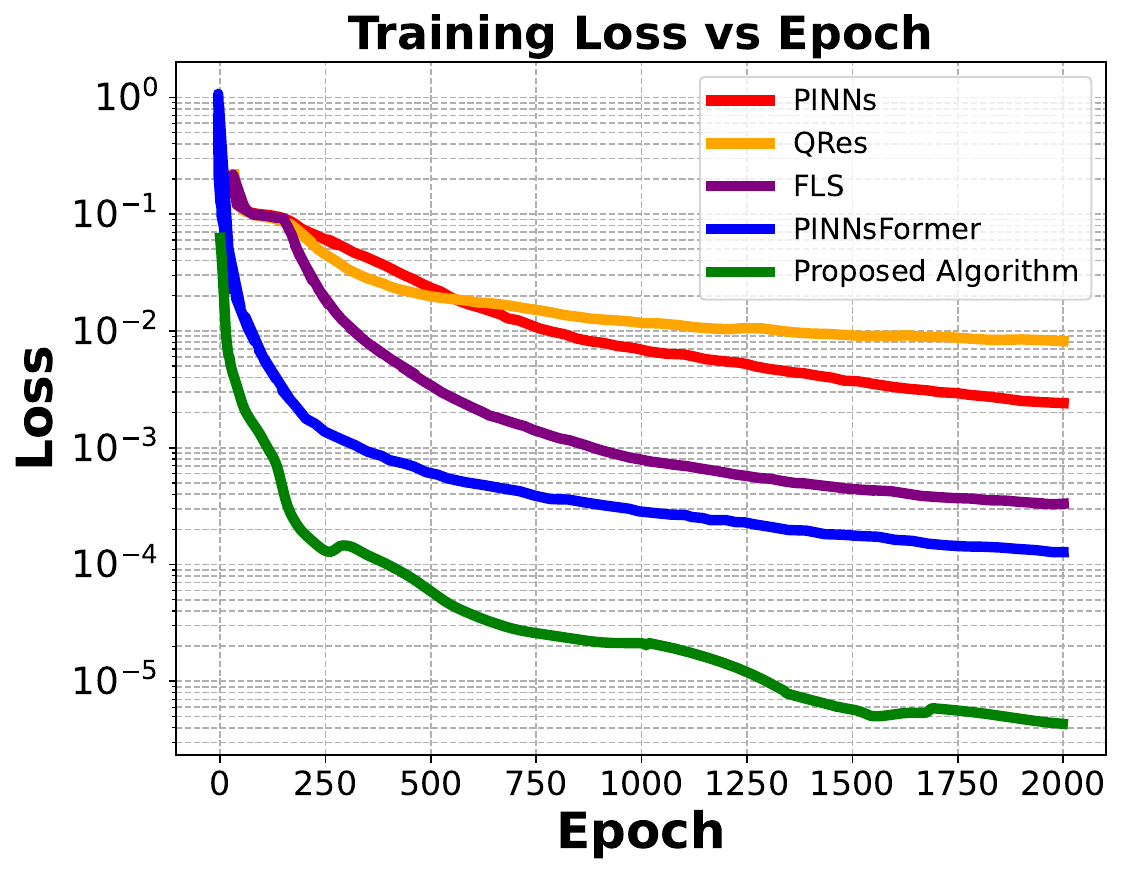}
\end{center}

\noindent\textbf{Figure \thefigure.}
Comparison of training loss convergence for \textit{PhysicsFormer}, PINNsFormer~\cite{zhao2023pinnsformer, zhaopinnsformer}, PINNs~\cite{raissi2019physics}, QRes~\cite{bu2021quadratic}, and FLS~\cite{wong2022learning} in the flow reconstruction problem at $Re=100$.

\label{fig:training_loss}

\refstepcounter{figure}

\begin{center}
    \includegraphics[width=\textwidth]{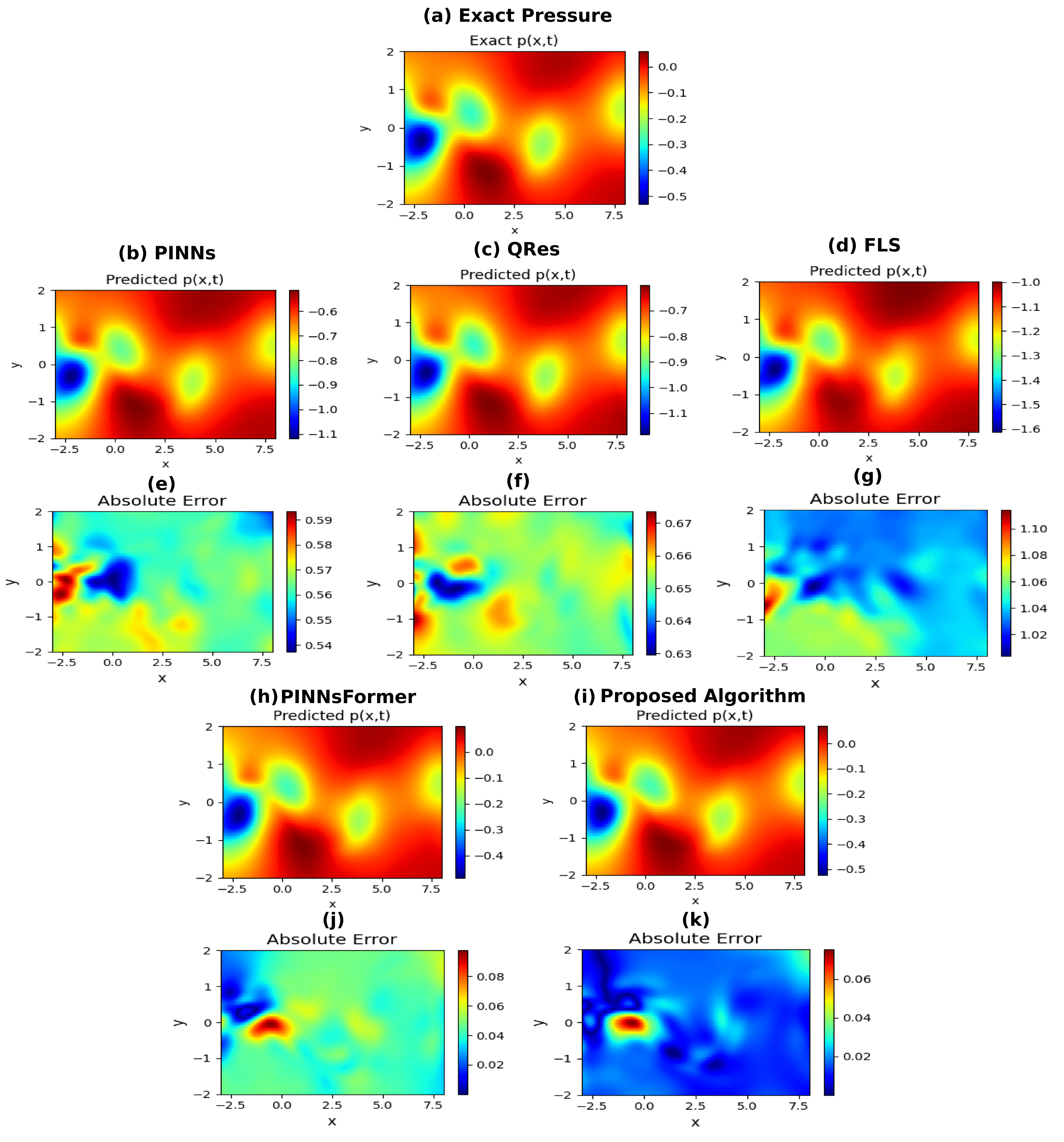}
\end{center}

\begin{quote}
\noindent\textbf{Figure \thefigure.}
Pressure field reconstruction for flow past a circular cylinder at $Re=100$. Panel (a) shows the exact pressure field~\cite{raissi2019physics}; panels (b), (c), (d), (h), and (i) present the predicted pressure fields obtained using PINNs~\cite{raissi2019physics}, QRes~\cite{bu2021quadratic}, FLS~\cite{wong2022learning}, PINNsFormer~\cite{zhao2023pinnsformer, zhaopinnsformer}, and the proposed \textit{PhysicsFormer}, respectively; while panels (e), (f), (g), (j), and (k) show the corresponding absolute errors.
\end{quote}

\label{fig:comparison_pdf}

\refstepcounter{figure}

\begin{center}
    \includegraphics[width=0.9\textwidth]{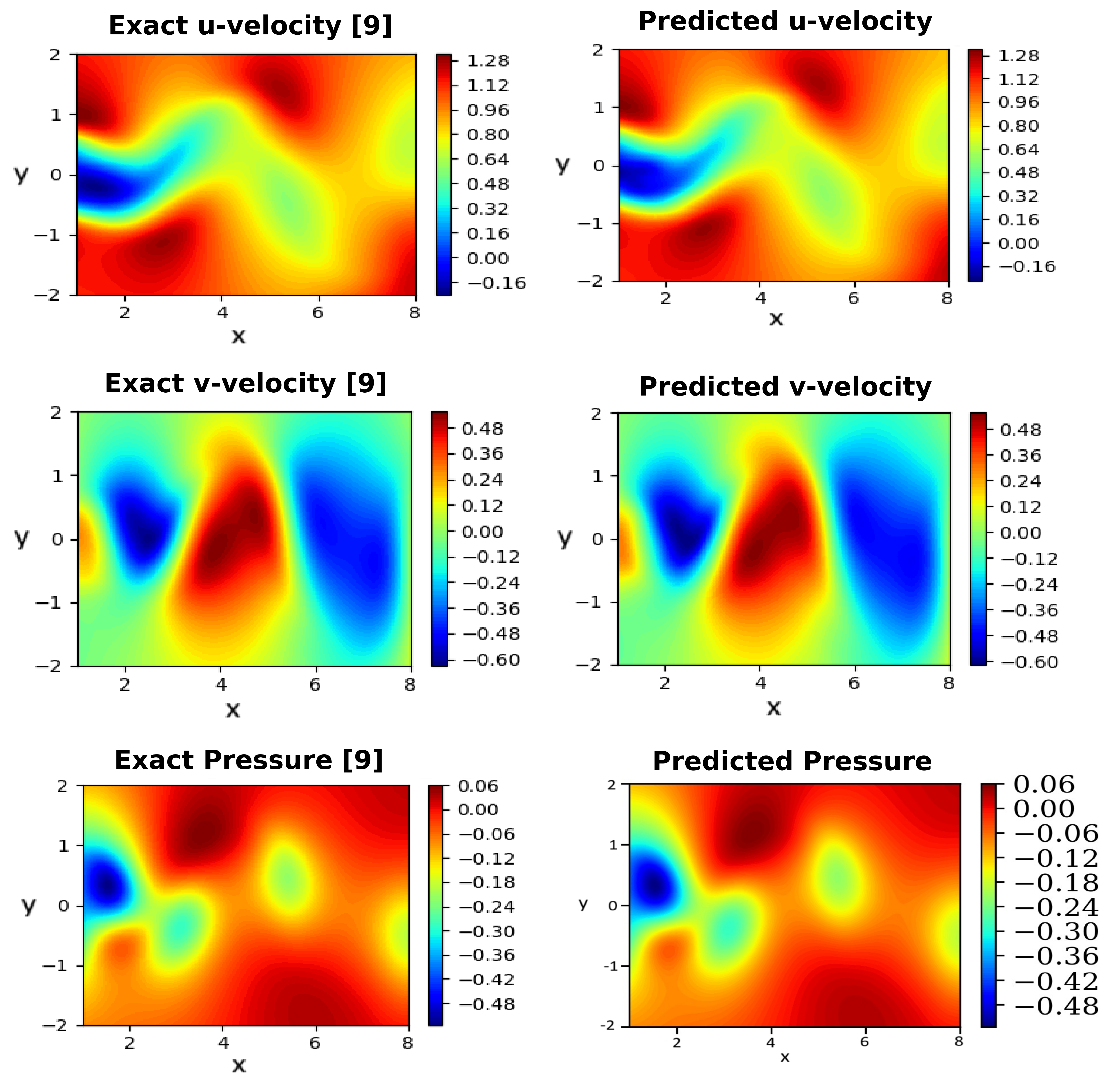}
\end{center}

\noindent\textbf{Figure \thefigure.}
Comparison of the reference solution~\cite{raissi2019physics} and the proposed \textit{PhysicsFormer} predictions for the $u$-velocity, $v$-velocity, and pressure fields at $Re=100$. The reconstruction uses only $0.15\%$ ($1500$) sparse supervised data, while accurately capturing the essential flow dynamics.

\label{fig:uvp_vs_physicsformer}

\refstepcounter{figure}

\begin{center}
    \includegraphics[width=0.8\textwidth]{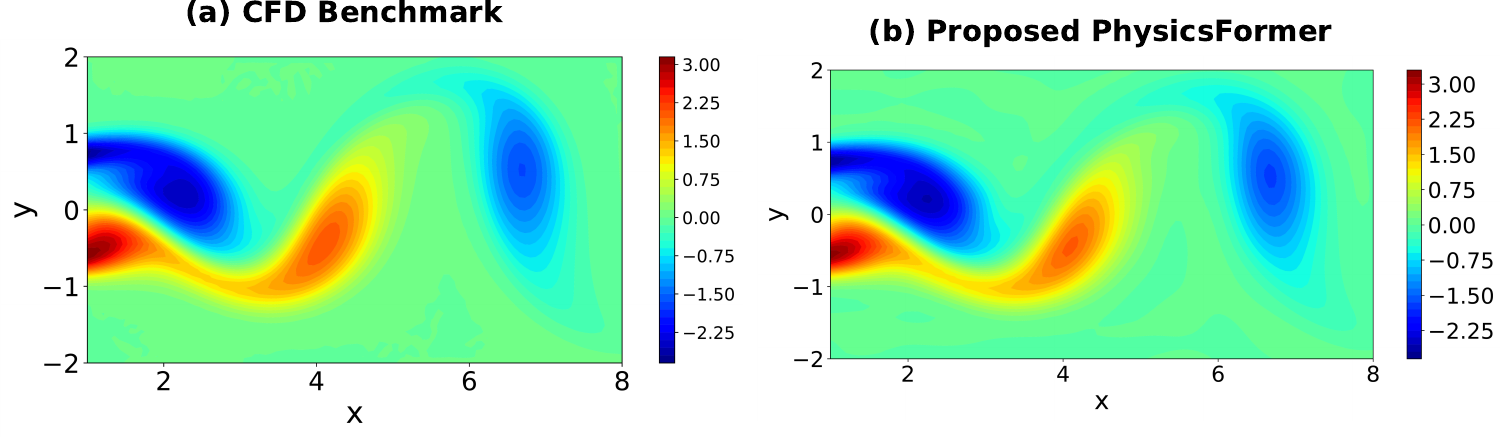}
\end{center}

\noindent\textbf{Figure \thefigure.}
Comparison of vorticity reconstruction in the wake of a circular cylinder at $Re=100$: (a) CFD benchmark solution~\cite{raissi2019physics} and (b) prediction from the proposed \textit{PhysicsFormer}. Using only $0.15\%$ ($1500$ samples) sparse data, the model accurately captures the wake vortex structures.

\label{fig:cfd_vs_physicsformer_vorticity}

\refstepcounter{figure}

\begin{center}
    \includegraphics[width=0.8\linewidth]{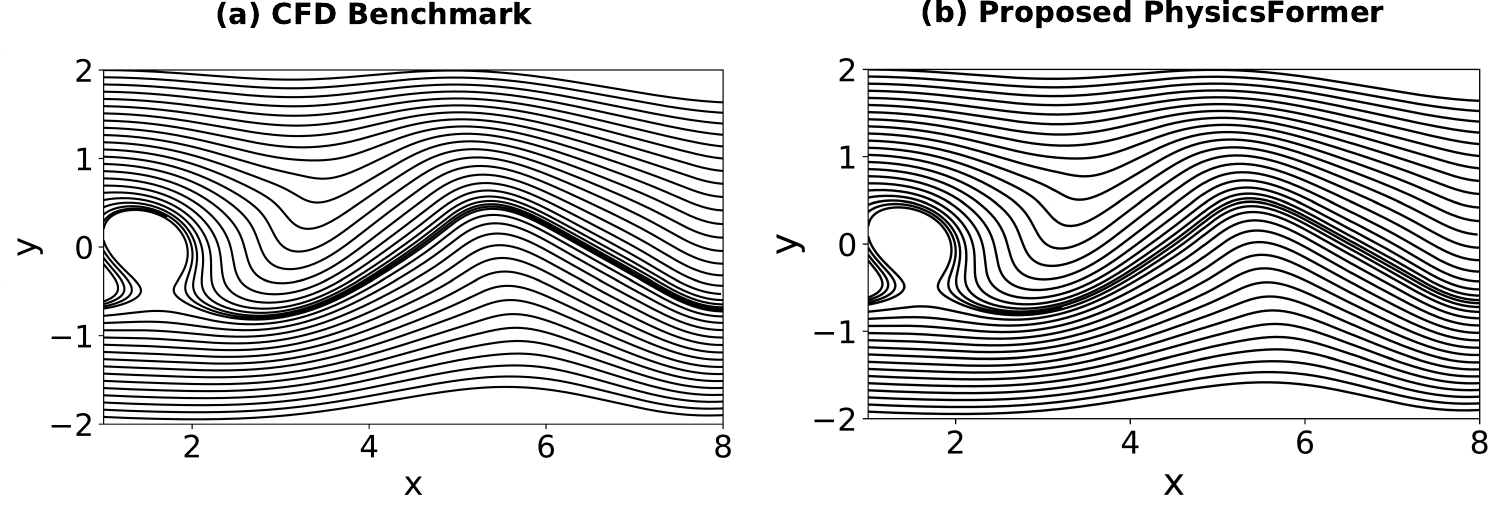}
\end{center}

\noindent\textbf{Figure \thefigure.}
Comparison of streamline reconstruction in the wake of a circular cylinder at $Re=100$: the left column shows the CFD benchmark solution~\cite{raissi2019physics}, while the right column presents the proposed \textit{PhysicsFormer} prediction. Using only $0.15\%$ ($1500$ samples) sparse velocity data, the model accurately captures the dominant wake flow structures.

\label{fig:streamline_comparison}

\begin{table}[htbp]
    \centering
    \caption{Comparison of training time and GPU memory consumption for Navier--Stokes flow reconstruction at $Re=100$ using an NVIDIA L4 GPU ($k=5$, $\Delta t = 1\times10^{-2}$).}
    \label{tab:time_memory}
    \begin{tabular}{lcc}
        \hline
        \textbf{Model Description} & \textbf{Duration of Training (seconds per epoch)} & \textbf{GPU Memory (MiB)} \\
        \hline
        PINNsFormer \cite{zhao2023pinnsformer}          & 11.04 & 2827 \\
        \textbf{Proposed Algorithm} & \textbf{3.60} & \textbf{640.5} \\
        \hline
    \end{tabular}
\end{table}

\clearpage
\subsection{Inverse Problem (Identified parameter and Equation)}
This work examines the inverse modeling of incompressible fluid flows governed by the two-dimensional Navier–Stokes equations. These equations are fundamental in characterizing numerous physical processes in research and engineering, including atmospheric circulation, ocean currents, aerodynamic design, hemodynamics, and pollution transport.  In two spatial dimensions, the Navier–Stokes equations can be expressed as
\begin{align}
u_t+\lambda_1\left(u u_x+v u_y\right) &= -p_x+\lambda_2\left(u_{x x}+u_{y y}\right), \label{eq:NS1}\\
v_t+\lambda_1\left(u v_x+v v_y\right) &= -p_y+\lambda_2\left(v_{x x}+v_{y y}\right), \label{eq:NS2}
\end{align}
where $u(t,x,y)$ and $v(t,x,y)$ denote the velocity components in the $x$- and $y$-directions, respectively, while $p(t,x,y)$ represents the pressure. The unknown physical parameters are given by $\lambda=(\lambda_1,\lambda_2)$, where $\lambda_1$ corresponds to the convection coefficient and $\lambda_2$ to the diffusion or viscosity coefficient.  

To ensure incompressibility, the velocity field must satisfy the continuity constraint
\begin{equation}
u_x + v_y = 0 .
\label{eq:continuity}
\end{equation}
This condition is enforced automatically by introducing a latent stream function $\psi(t,x,y)$ such that
\begin{equation}
u = \psi_y, \quad v = -\psi_x .
\end{equation}
Under this formulation, the divergence-free condition is satisfied by construction. Given noisy measurements 
$$
\{t^i, x^i, y^i, u^i, v^i\}_{i=1}^N ,
$$

Our aim is to jointly recover the latent stream function $\psi$, the pressure field $p$, vorticity, and the unknown parameters $\lambda=(\lambda_1,\lambda_2)$.

We provide \emph{PhysicsFormer}, a transformer-based model informed by physics, designed for spatio-temporal partial differential equations, as an alternative to traditional physics-informed neural networks (PINNs). The architecture is designed to capture periodic dynamics, temporal correlations, and long-range dependencies that are commonly seen in nonlinear fluid flows.

\subsubsection*{Architecture overview}
\begin{itemize}
    \item \textbf{Input encoding:} 
    The spatial-temporal inputs $[x,y,t]$ are transformed into a latent representation of dimension $d_{\text{model}}=128$ by a linear projection.
    
    \item \textbf{Transformer Architecture:} 
    The network utilizes an encoder-decoder architecture of $N=2$ layers and $4$ self-attention heads.  The attention mechanism facilitates the concurrent modeling of local flow interactions and global temporal dependencies.

    \item \textbf{Temporal learning:} 
    The data are transformed into sequences of length $\text{num\_step}=2$ with $\Delta t=10^{-4}$, enabling the model to approximate the temporal dynamics of the PDE solution.
    
    \item \textbf{Weighted sine activation:} 
    Each feedforward module employs a parameterized sine activation defined as $\phi(\mathrm{t})=w\sin(\mathrm{t})$, which improves the model's capacity to catch oscillatory patterns in contrast to conventional ReLU or Tanh functions.
    
    \item \textbf{Outputs:} 
    The decoder forecasts two scalar fields, $[\psi(t,x,y),p(t,x,y)]$, from which the velocity components $(u,v)$ can be derived.
      
    \item \textbf{Learnable PDE parameters:} 
    The coefficients $\lambda_1$ and $\lambda_2$ are regarded as trainable variables, facilitating inverse parameter identification in conjunction with field reconstruction.  
\end{itemize}

\refstepcounter{figure}

\begin{center}
    \includegraphics[width=0.9\textwidth,page=1]{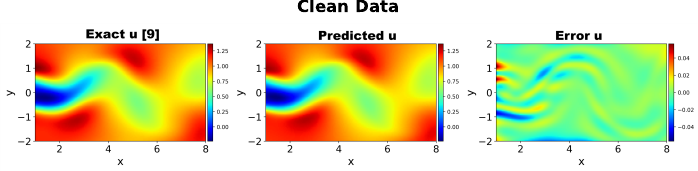}\\[0.2cm]
    \includegraphics[width=0.9\textwidth,page=2]{PhysicsFormer_inverse_NS_equation.pdf}\\[0.2cm]
    \includegraphics[width=0.9\textwidth,page=3]{PhysicsFormer_inverse_NS_equation.pdf}\\[0.2cm]
    \includegraphics[width=0.9\textwidth,page=4]{PhysicsFormer_inverse_NS_equation.pdf}
\end{center}

\noindent\textbf{Figure \thefigure.}
Inverse Navier--Stokes reconstruction at $Re=100$ using $1500$ supervised velocity samples. The top two rows show the exact~\cite{raissi2019physics} and \textit{PhysicsFormer} predictions for the $u$- and $v$-velocity fields with corresponding absolute errors for clean data, while the bottom two rows present the results under $1\%$ Gaussian noise.

\label{fig:velocity_comparison}

\refstepcounter{figure}

\begin{center}
    \includegraphics[width=0.9\linewidth]{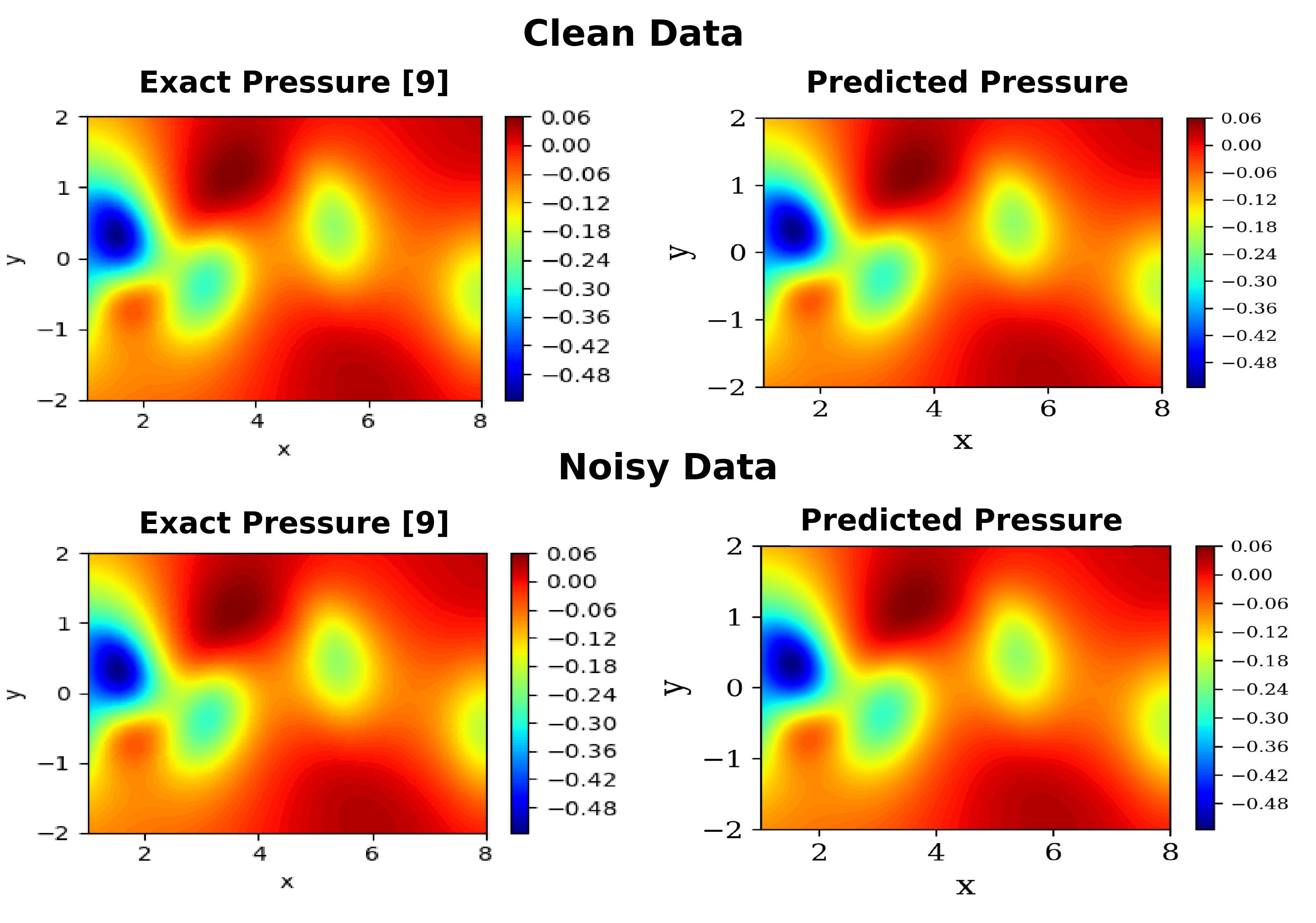}
\end{center}

\noindent\textbf{Figure \thefigure.}
Inverse Navier--Stokes pressure reconstruction at $Re=100$ using $1500$ supervised velocity samples. The top row compares the exact pressure field~\cite{raissi2019physics} and the \textit{PhysicsFormer} prediction for clean data, while the bottom row presents the corresponding results under $1\%$ Gaussian noise.\\

\label{fig:pressure_comparison}










To integrate the governing equations into the training process, \textit{PhysicsFormer} calculates the residuals associated with the momentum equations:

\begin{equation}
\begin{aligned}
f_u &:= u_t+\lambda_1\left(u u_x+v u_y\right)+p_x-\lambda_2\left(u_{x x}+u_{y y}\right), \\
f_v &:= v_t+\lambda_1\left(u v_x+v v_y\right)+p_y-\lambda_2\left(v_{x x}+v_{y y}\right).
\end{aligned}
\label{eq:fg_equations}
\end{equation}

The network is therefore assigned the twin objective of concurrently approximating the velocity fields $(u,v)$ via the stream function $\psi$ and the physics residuals $(f_u,f_v)$.

The comprehensive loss function is designed to ensure both conformity to the observed data and compliance with physical rules. The mean-squared error (MSE) loss is explicitly stated as

\begin{equation}
MSE := \frac{1}{\mathcal{N}} \sum_{i=1}^\mathcal{N}
\Big( |u(t^i,x^i,y^i) - u^i|^2 + |v(t^i,x^i,y^i) - v^i|^2 \Big)
+ \frac{1}{\mathcal{N}} \sum_{i=1}^\mathcal{N}
\Big( |f_u(t^i,x^i,y^i)|^2 + |f_v(t^i,x^i,y^i)|^2 \Big)
\end{equation}

The training dataset comprises $\mathcal{N}_{\text{train}}=1500$ spatio-temporal samples. Two scenarios are examined: one including noise-free data and the other incorporating $1\%$ Gaussian noise in the velocity measurements.  The network, with a hidden dimension of $d_{\text{hidden}}=256$, comprises several million parameters while maintaining memory efficiency owing to the brief sequence duration. Training occurs on CUDA-capable GPUs, with memory utilization assessed by monitoring allocated and reserved GPU resources. The architecture of this challenge is detailed in Table~\ref{tab:physicsformer_hyperparams}.

To assess the effectiveness of the proposed \textit{PhysicsFormer}, we compare our results with the reference data presented in Figures~\ref{fig:velocity_comparison} and \ref{fig:pressure_comparison}. The predicted velocity and pressure fields show strong agreement with the reference solutions for clean data, and the model remains robust even in the presence of $1\%$ Gaussian noise, with predictions closely matching the exact fields. The vortex-shedding dynamics in the wake region are further analyzed through the evolution of the vorticity field over a complete cycle. The shedding process originates at the initial time $t_0$, undergoes phase reversal at $t_0 + T/2$, and returns to an equivalent configuration at $t_0 + T$, where $T$ denotes the shedding period, as illustrated in Figure~\ref{fig:vorticity_cycle}. This behavior confirms that the proposed framework successfully captures the periodic vortex dynamics with high fidelity. Overall, these results indicate that \textit{PhysicsFormer} is capable of accurately reconstructing complex flow structures from sparse velocity measurements. In addition, we validate the learned parameters $\lambda_1$ and $\lambda_2$ against baseline approaches such as PINNs~\cite{raissi2019physics} and Res-PINNs~\cite{cheng2021deep}. The corresponding identification results are reported in Table~\ref{tab:lambda_comparison}, while the equation discovery performance is summarized in Table~\ref{tab:ns_equations}, both demonstrating the effectiveness of the proposed approach.

Figure~\ref{fig:lambda_convergence} illustrates the identification parameters $\lambda_1$ and $\lambda_2$, detailing their convergence during the training process. By the last training epoch, $5000$, both parameters closely correspond with the true value.

\refstepcounter{figure}

\begin{center}
    \includegraphics[width=1.0\linewidth]{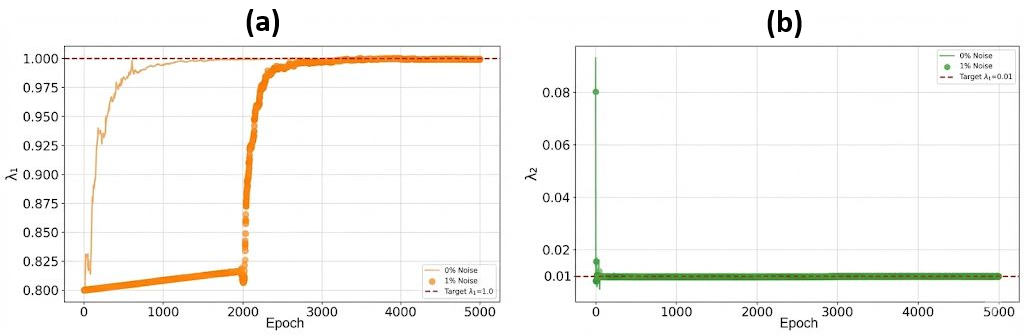}
\end{center}

\noindent\textbf{Figure \thefigure.}
Convergence histories of the identified parameters $\lambda_1$ and $\lambda_2$ for the inverse Navier--Stokes problem at $Re=100$ over 5000 epochs under clean and noisy data conditions. Panels (a) and (b) show the convergence of $\lambda_1$ and $\lambda_2$, respectively.

\label{fig:lambda_convergence}

\refstepcounter{figure}

\begin{center}
    \includegraphics[width=1.0\textwidth]{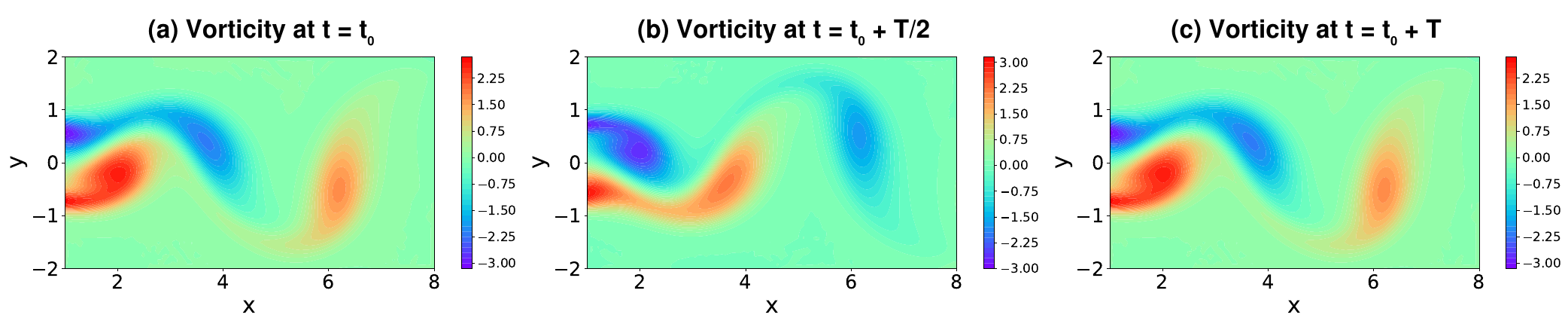}
\end{center}

\noindent\textbf{Figure \thefigure.}
Temporal evolution of vorticity fields for flow past a circular cylinder at $Re=100$ over one complete vortex-shedding cycle: (a) $\mathrm{t}=\mathrm{t}_0$, (b) $\mathrm{t}=\mathrm{t}_0 + T/2$, and (c) $\mathrm{t}=\mathrm{t}_0 + T$, where $T$ denotes the full shedding period.

\label{fig:vorticity_cycle}

\begin{table}[htbp]
\centering
\caption{Hyperparameter settings of \textit{PhysicsFormer} for the inverse Navier--Stokes problem at $Re=100$.}
\label{tab:physicsformer_hyperparams}
\begin{tabular}{ll}
\hline
\textbf{Component} & \textbf{Setting} \\
\hline
Input Encoding & $[x,y,t] \to d_{\text{model}}=128$ \\
Optimizer & L-BFGS \\
Epochs & $5000$ \\
Activation & Weighted \textit{Sine} \\
\hline
\multicolumn{2}{c}{\textbf{Transformer Architecture}} \\
\hline
Architecture & Encoder--decoder \\
Layers & $N=2$ \\
Attention Heads & $4$ \\
\hline
\multicolumn{2}{c}{\textbf{Temporal Sequence}} \\
\hline
Length & $\text{num\_step}=2$ \\
Interval & $\Delta t = 10^{-4}$ \\
\hline
\multicolumn{2}{c}{\textbf{Training Data}} \\
\hline
Samples & $N_{\text{train}}=1500$ \\
Noise & 0\% and 1\% Gaussian \\
\hline
\end{tabular}
\end{table}

\begin{table}[htbp]
\centering
\caption{Comparison of the identified parameters $\lambda_1$ and $\lambda_2$ for the inverse Navier--Stokes problem at $Re=100$ using PINNs~\cite{raissi2019physics}, Res-PINN~\cite{cheng2021deep}, and the proposed \textit{PhysicsFormer} under clean and $1\%$ noisy data conditions.}
\label{tab:lambda_comparison}
\begin{tabular}{lcccccccc}
\hline
\multirow{2}{*}{Model} & \multicolumn{4}{c}{Clean Data} & \multicolumn{4}{c}{1\% Noisy Data} \\
\cline{2-5} \cline{6-9}
 & $\lambda_1$ & Err(\%) & $\lambda_2$ & Err(\%) & $\lambda_1$ & Err(\%) & $\lambda_2$ & Err(\%) \\
\hline
PINNs         & 0.999 & 0.07  & 0.01047 & 4.67  & 0.998 & 0.17  & 0.01057 & 5.70  \\
Res-PINN \cite{cheng2021deep}      & 1.000 & 0.0 & 0.01006 & 0.61  & 1.000 & 0.0 & 0.01011 & 1.08  \\
\textbf{Proposed Algorithm} & \textbf{1.000} & \textbf{0.0} & \textbf{0.0100} & \textbf{0.0}  & 0.999 & 0.07  & \textbf{0.01000} & \textbf{0.0}  \\
\hline
\end{tabular}
\end{table}

\begin{table}[htbp]
\centering
\caption{Comparison between the exact and identified Navier--Stokes equations for the inverse problem at $Re=100$ using the proposed \textit{PhysicsFormer}.}
\label{tab:ns_equations}
\renewcommand{\arraystretch}{1.3} 
\begin{tabular}{p{0.42\textwidth} p{0.52\textwidth}}
\hline
\textbf{Corrected Navier-Stokes equation} & 
$\begin{aligned}
u_t + (uu_x + vu_y) &= -p_x + 0.01(u_{xx} + u_{yy}), \\
v_t + (uv_x + vv_y) &= -p_y + 0.01(v_{xx} + v_{yy}),
\end{aligned}$ \\
\hline

Identified Navier-Stokes equation using clean data (PINN)  &
$\begin{aligned}
u_t + 0.999(uu_x + vu_y) &= -p_x + 0.01047(u_{xx} + u_{yy}), \\
v_t + 0.999(uv_x + vv_y) &= -p_y + 0.01047(v_{xx} + v_{yy}),
\end{aligned}$ \\
\hline

Identified Navier-Stokes equation using clean data (Res-PINN) \cite{cheng2021deep}&
$\begin{aligned}
u_t + 1.000(uu_x + vu_y) &= -p_x + 0.01006(u_{xx} + u_{yy}), \\
v_t + 1.000(uv_x + vv_y) &= -p_y + 0.01006(v_{xx} + v_{yy}),
\end{aligned}$ \\
\hline

\textbf{Identified Navier-Stokes equation using clean data(Proposed Algorithm)} &
$\begin{aligned}
u_t + \textbf{1.000}(uu_x + vu_y) &= -p_x + \textbf{0.01000}(u_{xx} + u_{yy}), \\
v_t + \textbf{1.000}(uv_x + vv_y) &= -p_y + \textbf{0.01000}(v_{xx} + v_{yy}),
\end{aligned}$ \\
\hline


Identified Navier-Stokes equation using noise data (PINN) &
$\begin{aligned}
u_t + 0.998(uu_x + vu_y) &= -p_x + 0.01057(u_{xx} + u_{yy}), \\
v_t + 0.998(uv_x + vv_y) &= -p_y + 0.01057(v_{xx} + v_{yy}),
\end{aligned}$ \\
\hline

Identified Navier-Stokes equation using noise data (Res-PINN) \cite{cheng2021deep}&
$\begin{aligned}
u_t + 1.000(uu_x + vu_y) &= -p_x + 0.01011(u_{xx} + u_{yy}), \\
v_t + 1.000(uv_x + vv_y) &= -p_y + 0.01011(v_{xx} + v_{yy}),
\end{aligned}$ \\
\hline

\textbf{Identified Navier-Stokes equation using noise data (Proposed Algorithm)} &
$\begin{aligned}
u_t + 0.999(uu_x + vu_y) &= -p_x + \textbf{0.01000}(u_{xx} + u_{yy}), \\
v_t + 0.999(uv_x + vv_y) &= -p_y + \textbf{0.01000}(v_{xx} + v_{yy}),
\end{aligned}$ \\
\hline
\end{tabular}
\end{table}

\clearpage
\subsection{Flow Reconstruction for Flow Past a Circular Cylinder at $Re=3900$}

To further evaluate the capability of the proposed PhysicsFormer framework for complex unsteady flows, we consider the flow reconstruction problem for the wake flow past a two-dimensional circular cylinder at Reynolds number $Re=3900$. This benchmark problem is characterized by highly nonlinear vortex shedding dynamics and has been widely used to assess reduced-order and data-driven flow modeling approaches. The flow field is generated using the Reynolds-averaged Navier--Stokes (RANS) equations with the $k$--$\omega$~\cite{yan2023exploring} turbulence model.

The governing incompressible Navier--Stokes equations are expressed as
\begin{align}
\frac{\partial u}{\partial x} + \frac{\partial v}{\partial y} &= 0, \\
\frac{\partial u}{\partial t}
+ u\frac{\partial u}{\partial x}
+ v\frac{\partial u}{\partial y}
+ \frac{\partial p}{\partial x}
- \frac{1}{Re}
\left(
\frac{\partial^2 u}{\partial x^2}
+
\frac{\partial^2 u}{\partial y^2}
\right)
&= 0, \\
\frac{\partial v}{\partial t}
+ u\frac{\partial v}{\partial x}
+ v\frac{\partial v}{\partial y}
+ \frac{\partial p}{\partial y}
- \frac{1}{Re}
\left(
\frac{\partial^2 v}{\partial x^2}
+
\frac{\partial^2 v}{\partial y^2}
\right)
&= 0,
\end{align}
where $u$ and $v$ denote the velocity components in the streamwise and transverse directions, respectively, while $p$ represents the pressure field.

For training, sparse flow measurements consisting of $(u,v,p)$ are collected from only 25 spatially distributed sensor locations over a time interval of $42.9\,\mathrm{s}$. The dataset contains 100 temporal snapshots, resulting in a total of 2500 labeled data points. Despite the extremely limited observations, the proposed PhysicsFormer successfully reconstructs the complete velocity and pressure fields with high accuracy.

The model employs an encoder--decoder transformer architecture with two layers and four attention heads. The spatial-temporal coordinates $(x,y,t)$ are embedded into a latent representation of dimension $d_{\mathrm{model}}=512$. A weighted sine activation function is adopted to improve the representation of complex flow dynamics. The network is trained using the L-BFGS optimizer for 5000 epochs. For temporal sequence modeling, a sequence length of $\text{num\_step}=2$ with pseudo-time interval $\Delta t = 10^{-4}$ is used.

The flow reconstruction capability of \textit{PhysicsFormer} is investigated for the flow past a circular cylinder at $Re=3900$ using only $25$ supervised spatial measurements across $100$ temporal snapshots. The reconstructed $u$-velocity, $v$-velocity, and pressure fields are evaluated at two representative time instants: an early-time snapshot at $t=0.43\,\mathrm{s}$ (Figure.~\ref{fig:re3900_uvp_results}) and a long-time snapshot at $t=17.78\,\mathrm{s}$ (Figure.~\ref{fig:re3900_flow_reconstruction}). In both cases, the reconstructed flow fields show strong agreement with the reference solutions, demonstrating that \textit{PhysicsFormer} effectively captures complex wake dynamics and long-range temporal dependencies from highly sparse observations.

Quantitatively, PhysicsFormer attains relative errors of $4.8\times10^{-3}$, $1.6\times10^{-2}$, and $2.0\times10^{-2}$ for the streamwise velocity, transverse velocity, and pressure fields, respectively. These values are notably lower than those reported by the STPINNs-RANS framework~\cite{xu2023spatiotemporal}, which yields errors of $1.8\times10^{-2}$, $3.3\times10^{-2}$, and $4.0\times10^{-2}$ for the corresponding fields. As summarized in Table~\ref{tab:cylinder_reconstruction_comparison}, this comparison indicates that PhysicsFormer provides improved reconstruction accuracy while maintaining computational efficiency for high-Reynolds-number turbulent flow problems.

\begin{table}[htbp]
\centering
\caption{Comparison of relative RMSE(rRMSE) for flow reconstruction of the flow past a circular cylinder at $Re=3900$ between STPINNs-RANS~\cite{xu2023spatiotemporal} and the proposed PhysicsFormer framework.}
\label{tab:cylinder_reconstruction_comparison}
\begin{tabular}{lccc}
\hline
\textbf{Model} & $\mathbf{u}$-velocity & $\mathbf{v}$-velocity & \textbf{Pressure} \\
\hline
STPINNs-RANS~\cite{xu2023spatiotemporal} & 0.018 & 0.033 & 0.040 \\
\textbf{PhysicsFormer} & \textbf{0.0048} & \textbf{0.016} & \textbf{0.020} \\
\hline
\end{tabular}
\end{table}

\refstepcounter{figure}

\begin{center}
    \includegraphics[width=0.9\textwidth]{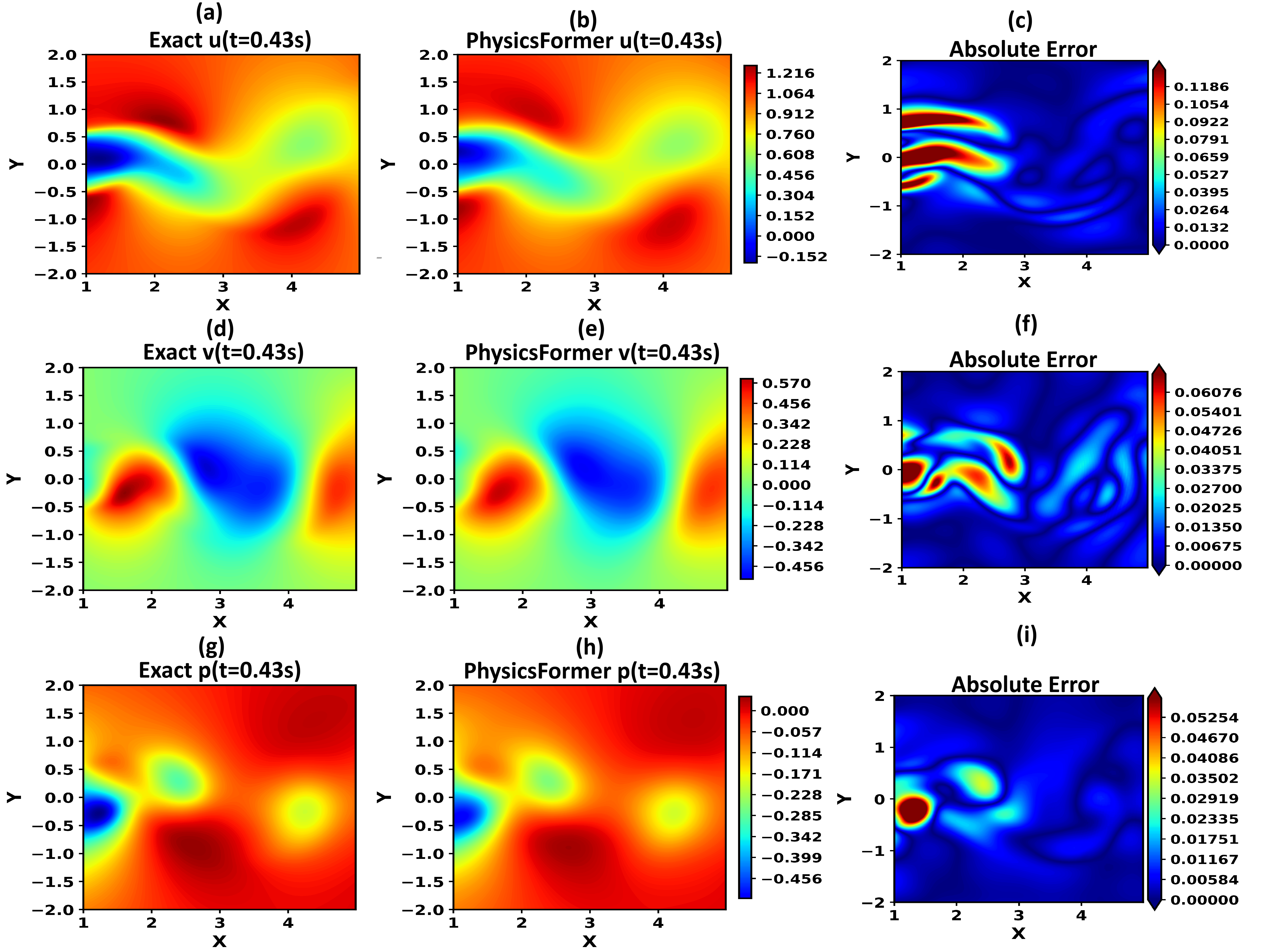}
\end{center}

\noindent\textbf{Figure \thefigure.}
Flow reconstruction results for flow past a circular cylinder at $Re=3900$ at $t=0.43\,\mathrm{s}$. Panels (a), (d), and (g) show the exact $u$-velocity, $v$-velocity, and pressure fields; panels (b), (e), and (h) present the corresponding \textit{PhysicsFormer} predictions; and panels (c), (f), and (i) illustrate the associated absolute errors.

\label{fig:re3900_uvp_results}

\refstepcounter{figure}

\begin{center}
    \includegraphics[width=0.9\textwidth]{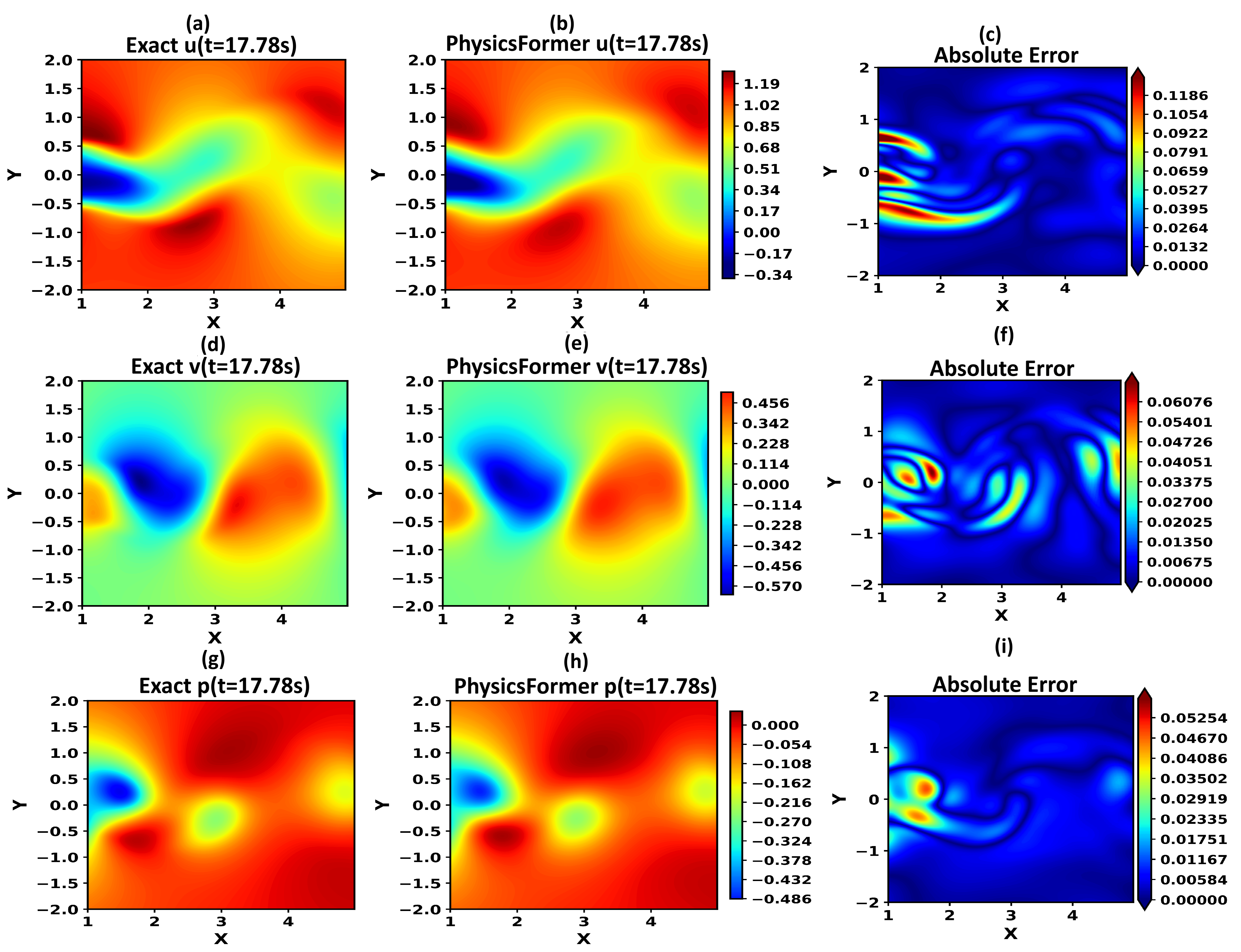}
\end{center}

\noindent\textbf{Figure \thefigure.}
Flow reconstruction results for flow past a circular cylinder at $Re=3900$ at $t=17.78\,\mathrm{s}$. Panels (a), (d), and (g) show the exact $u$-velocity, $v$-velocity, and pressure fields, respectively; panels (b), (e), and (h) present the corresponding \textit{PhysicsFormer} predictions; and panels (c), (f), and (i) illustrate the associated absolute errors.

\label{fig:re3900_flow_reconstruction}

\clearpage
\renewcommand{\thesection}{\arabic{section}}
\setcounter{section}{5}
\section{Conclusion}
\label{sec:conclusion}

This work presents \textit{PhysicsFormer}, a lightweight and efficient Transformer-based physics-informed neural network for surrogate modeling of nonlinear fluid dynamics and inverse parameter identification. The study first investigates the failure modes of conventional PINNs using a simple convection problem, highlighting their limitations in capturing temporally evolving dynamics. To address these challenges, the proposed framework is systematically validated on a sequence of benchmark problems, including the Burgers' equation, lid-driven cavity flow at $Re=100$, flow reconstruction, and inverse Navier--Stokes parameter identification for flow past a circular cylinder at $Re=100$, followed by high-Reynolds-number flow reconstruction at $Re=3900$. Across all forward reconstruction tasks, the proposed method accurately recovers velocity, pressure, vorticity, and streamline fields from sparse observations while preserving complex unsteady flow structures over long temporal horizons. In inverse problems, \textit{PhysicsFormer} reliably identifies the unknown Navier--Stokes parameters, achieving nearly $0\%$ absolute error in both clean and noisy data settings. The results demonstrate strong agreement with reference solutions across all test cases, including strongly nonlinear and high-Reynolds-number regimes. In addition, the proposed framework offers significant computational advantages over existing Transformer-based PINN models, providing faster training and efficient execution on standard GPU platforms. Overall, \textit{PhysicsFormer} provides a robust, accurate, and scalable approach for physics-informed deep learning in computational fluid dynamics, with strong potential for extension to more complex turbulent and multiphase flow systems.

\section*{Appendix A. Ablation investigation to determine \( k \) and \( \Delta t \) for the Burgers’ Equation}

Selecting k and $\Delta t$ is important in these research due to the sensitivity of both factors.  We consider two principal fluid flow equations: Burgers' equation and the flow reconstruction and inverse issue guided by the Navier-Stokes equation.  Our proposed \textit{PhysicsFormer} resembles a grid-based approach, where \( k \) represents the number of steps and \( \Delta t \) signifies the step size. If we increase the number of steps, our model occupies huge GPU memory and incurs significant computational overhead. Therefore, we standardize the number of steps to \( k=5 \) for all problems.  Conversely, if the step size is excessively tiny, the accuracy of the results may improve, although it requires tremendous processing cost in these analyses.  In the flow reconstruction problem, we set $\Delta t=1 \times 10^{-2}$, while in the inverse problem, we designate $\Delta t=1 \times 10^{-4}$.  In relation to the Burgers' equation, we examine the values $k=4, 5, \& 6$ and $\Delta t= 1 \times 10^{-3}, 1 \times 10^{-4} \& 1 \times 10^{-5}$ as presented in Table~\ref{tab:k_dt_burgers}.  We utilize several combinations while maintaining consistent hyperparameters across all instances Table~\ref{tab:physicformer_burger}. For the purpose of assessing model correctness, we evaluate the MAE, RMSE, relative RMSE, and Loss.  We observed that for \( k=5 \) and \( \Delta t=1 \times 10^{-4} \), there is significant agreement because of lower MAE, RMSE, relative RMSE and Loss; so, we keep these values in the equations. In higher-dimensional problems these combinations typically perform effectively; however, they require experimentation similar to other network hyperparameters. In the flow reconstruction problem, both $k=5$ and $\Delta t= 1 \times 10^{-2}$ and $1 \times 10^{-4}$ yielded comparable results.  Therefore, we set $k=5$ and $\Delta = 1 \times 10^{-2}$, whereas for the other problem, we maintain $k=5$ and $\Delta t = 1 \times 10^{-4}$.

\begin{table}[htbp]
\centering
\caption{Performance analysis of Burgers' equation for various combinations of \(k\) and \(\Delta t\), keeping all other hyperparameters constant.}
\begin{tabular}{cccccc}
\hline
\textbf{num\_step (\(k\))} & \textbf{step (\(\Delta t\))} & \textbf{MAE} & \textbf{RMSE} & \textbf{rRMSE} & \textbf{Loss} \\
\hline
4 & \(1\times10^{-3}\) & 0.002085 & 0.012855 & 0.021033 & 6.018963$\times10^{-6}$ \\
4 & \(1\times10^{-4}\) & 0.002096 & 0.013567 & 0.022198 & 6.264328$\times10^{-6}$ \\
4 & \(1\times10^{-5}\) & 0.002017 & 0.013711 & 0.022434 & 6.585098$\times10^{-6}$ \\
5 & \(1\times10^{-3}\) & 0.002408 & 0.016597 & 0.027157 & 9.927394$\times10^{-6}$ \\
\textbf{5} & \(\mathbf{1\times10^{-4}}\) & \textbf{0.001539} & \textbf{0.009144} & \textbf{0.014962} & \textbf{1.262837$\times10^{-6}$} \\
5 & \(1\times10^{-5}\) & 0.003234 & 0.019685 & 0.032208 & 2.904948$\times10^{-6}$ \\
6 & \(1\times10^{-3}\) & 0.002589 & 0.016487 & 0.026977 & 6.659313$\times10^{-6}$ \\
6 & \(1\times10^{-4}\) & 0.002709 & 0.017549 & 0.028714 & 6.215727$\times10^{-6}$ \\
6 & \(1\times10^{-5}\) & 0.001904 & 0.011310 & 0.018506 & 1.752262$\times10^{-5}$ \\
\hline
\end{tabular}
\label{tab:k_dt_burgers}
\end{table}

\noindent
The accuracy of the predicted solution \(\hat{u}(x,t)\) is evaluated using the following error metrics:

\begin{align}
MAE &= \frac{1}{\mathcal{N}} \sum_{i=1}^{\mathcal{N}} \left| u_{i}^{\text{pred}} - u_{i}^{\text{true}} \right|, \label{eq:L1}\\[6pt]
RMSE &= \sqrt{\frac{1}{\mathcal{N}} \sum_{i=1}^{\mathcal{N}} \left| u_{i}^{\text{pred}} - u_{i}^{\text{true}} \right|^{2}}, \label{eq:L2}\\[6pt]
\text{relative}~RMSE~(rRMSE)~ &= \frac{ \sqrt{ \sum_{i=1}^{\mathcal{N}} \left| u_{i}^{\text{pred}} - u_{i}^{\text{true}} \right|^{2}} }{ \sqrt{ \sum_{i=1}^{\mathcal{N}} \left| u_{i}^{\text{true}} \right|^{2}} }, \label{eq:RelL2}
\end{align}

\noindent
where \( u_{i}^{\text{pred}} \) and \( u_{i}^{\text{true}} \) denote the predicted and reference solution values at the \( i^{\text{th}} \) grid point, respectively, and \(\mathcal{N}\) represents the total number of data points used for evaluation.

\section*{Appendix B. Ablation Study of Different Activation Functions for the Flow Reconstruction Problem}
Four distinct activation functions are used in the activation function ablation study: Sine, Tanh, Wavelet(x) = $\omega_1sin(x) + \omega_2cos(x)$, and the weighted Sine activation function. We maintain a similar architecture throughout in Table 2:  The only significant changes are to the number of heads and hidden nodes ($d_{\text{hidden}}$), which are now 2 and 512, respectively.  We compare the number of parameters, total computational cost, rMAE, rRMSE, and Loss for each of these activation functions in order to ensure accuracy.  Our findings are included in Table~\ref{tab:flow_reconstruction_activation}, where we can observe that the weighted Sine gradually reduces loss, error, and computation time.  These will require a lot of GPU RAM if the Wavelet activation function \cite{zhao2023pinnsformer} is utilized. Our proposed simplest modification of these activation functions, weighted Sine, runs well on the Google Colab T-4, 15GB GPU.

\appendix
\section*{Appendix C: Effect of the Number of Attention Heads}

\noindent
The total number of attention heads is essential in influencing the model's distribution of representational capacity across various subspaces of the embedded features. In a multi-head attention mechanism, the total embedding dimension, referred to as \( d_{\text{model}} \), is distributed among the attention heads. The dimensionality associated with each head is so formulated as
$$
d_{\text{heads}} = \frac{d_{\text{model}}}{N_{\text{heads}}}.
$$
Each attention head autonomously acquires distinct spatial-temporal correlations or fundamental physical linkages inherent in the flow field.  Augmenting the number of heads allows the model to concentrate on more localized dependencies; however, this concurrently reduces the representational dimensionality accessible to each head.  When the embedding dimension \( d_{\text{model}} \) is rather small, an excessive number of heads may result in a poor representation of essential flow dynamics, thereby impacting the model's convergence and overall stability.

\medskip
\noindent
This study set the embedding dimension at \( d_{\text{model}} = 32 \) and changed the number of attention heads \( (N_{\text{heads}}) \) at 2, 4, and 8 to evaluate their impact on predictive performance, as shown in Table~\ref{tab:attention_heads}.  Among these designs, the arrangement with $4$ attention heads yielded the most favorable performance in terms of loss, relative mean absolute error (rMAE), relative root mean squared error (rRMSE), and computational efficiency.  If one utilizes these high attention heads, there is a possibility that the model will overfit. All models are executed on the NVIDIA L4 GPU, however they are also compatible with the Google Colab T4 GPU.

\medskip
\noindent
The relative error metrics are mathematically defined as follows:
$$
\mathrm{rMAE} = 
\frac{\displaystyle \sum_{i=1}^{\mathcal{N}} 
\left| u_{i}^{\text{pred}} - u_{i}^{\text{true}} \right|}
{\displaystyle \sum_{i=1}^{\mathcal{N}} 
\left| u_{i}^{\text{true}} \right|},
\qquad
\mathrm{rRMSE} \;=\;
\sqrt{
\frac{\displaystyle \sum_{n=1}^{\mathcal{N}} \left| u_{i}^{\text{pred}} - u_{i}^{\text{true}} \right|^{2}}
     {\displaystyle \sum_{n=1}^{\mathcal{N}} \left| u_{i}^{\text{true}} \right|^{2}}
}\ 
$$
where \( u_{i}^{\text{pred}} \) and \( u_{i}^{\text{true}} \) denote the predicted and reference solution values at the \( i^{\text{th}} \) grid point, respectively.

\begin{table}[h!]
\centering
\caption{Activation function ablation study for the Flow Reconstruction problem.}
\begin{tabular}{lcccccc}
\hline
\textbf{Activation} & \textbf{Loss} & \textbf{rMAE} & \textbf{rRMSE} & \textbf{Model Parameters} & \textbf{Time (min)} & \textbf{GPU Name} \\
\hline
Sine          & 0.000014 & 1.055 & 0.728 & 454082 & 70 & L4, 24GB\\
Tanh          & 0.000984 & 1.335 & 0.929 & 454082 & 160 & L4, 24GB \\
Wavelet       & 0.000008 & 0.404 & 0.293 & 454106 & 184 & 
L4, 24GB\\
Weighted Sine & \textbf{0.000006} & \textbf{0.176} & \textbf{0.137} & 454094 & \textbf{80} & L4, 24GB \\
\hline
\end{tabular}
\label{tab:flow_reconstruction_activation}
\end{table}

\begin{table}[htbp]
\centering
\caption{Performance comparison for varying number of attention heads in the Flow Reconstruction problem 
(\(k = 5\), \( \Delta t = 1\times10^{-2} \)).}
\label{tab:attention_heads}
\begin{tabular}{ccccc}
\hline
\textbf{heads (\(N_{\text{heads}}\))} & \textbf{rMAE} & \textbf{rRMSE} & \textbf{Loss} & \textbf{Computational Time (min)} \\
\hline
2 & 0.398 & 0.292 & 1.02$\times10^{-5}$ & 40 \\
\textbf{4} & \textbf{0.136} & \textbf{0.133} & \textbf{5.35$\times10^{-6}$} & \textbf{60} \\
8 & 1.897 & 1.319 & 8.17$\times10^{-6}$ & 167 \\
\hline
\end{tabular}
\end{table}

\section*{Declaration of Generative AI and AI-assisted Technologies in the Writing Process}

During the preparation of this manuscript, the authors used ChatGPT solely for language polishing and grammatical correction. After using this tool, the authors carefully reviewed, revised, and validated the content. The authors take full responsibility for the scientific integrity and content of the published article.

\section*{Author Declaration}
The authors have no conflicts of interest to disclose.

\section*{Data Availability}
The data that support the findings of this study are available from the corresponding author upon reasonable request.

\end{document}